\documentclass{article}

\usepackage{arxiv}
\usepackage[dvipsnames]{xcolor}
\usepackage{graphicx}
\graphicspath{{img/}}
\usepackage{balance}  % for  \balance command ON LAST PAGE  (only there!)
\usepackage{amsmath}
\usepackage{amsthm}
\usepackage{amsfonts}
\usepackage{amssymb}
\usepackage{multirow}
\usepackage{array}
\usepackage{mathtools}
\usepackage{hyperref}
\usepackage{tikz}
\usepackage{color, colortbl}
\usetikzlibrary{cd, positioning}
\usepackage[inline]{enumitem}
\usepackage{listings}
\title{Sum-Product Optimization via
  Relational Equality Saturation for Large Scale Linear Algebra}
\newcommand{\lrparen}[1]{% \lrparen{..}
  \left(#1\right)}

\definecolor{Gray}{gray}{0.8}

\newcommand{\MPLAN}[0]{\textsc{LA}}

\newcommand{\rmul}[0]{*}
\newcommand{\radd}[0]{+}

\newcommand{\unbind}[2]{{\ensuremath{\scriptstyle\left[-#1,-#2\right]\,}}}
\newcommand{\bind}[1]{{\ensuremath{\scriptstyle\left[#1\right]\,}}}
\newcommand{\tr}[0]{{\mathsf{T}}}

\newtheorem{definition}{Definition}[section]
\newtheorem{theorem}{Theorem}[section]
\newtheorem{lemma}[theorem]{Lemma}
\newtheorem{example}{Example}
\newtheorem{corollary}{Corollary}

\newcommand{\sys}{SPORES}

\newcommand{\N}{\mathbb N} % the natural numbers
\newcommand{\R}{\mathbb R} % the real numbers
\begin{document}

% ****************** TITLE ****************************************

\title{\sys: Sum-Product Optimization via
  Relational Equality Saturation for Large Scale Linear Algebra}

% possible, but not really needed or used for PVLDB: \subtitle{[Extended
%Abstract] \titlenote{A full version of this paper is available
%as\textit{Author's Guide to Preparing ACM SIG Proceedings Using
%\LaTeX$2_\epsilon$\ and BibTeX} at \texttt{www.acm.org/eaddress.htm}}}

% ****************** AUTHORS **************************************


% \numberofauthors{5} %  in this sample file, there are a *total*
\author{
Yisu Remy Wang\\
      University of Washington\\
      Seattle, Washington\\
      \texttt{remywang@cs.washington.edu}
% 2nd. author
\And
Shana Hutchison\\
      University of Washington\\
      Seattle, Washington\\
      \texttt{shutchis@cs.washington.edu}
% 3rd. author
\And Jonathan Leang\\
      University of Washington\\
      Seattle, Washington\\
      \texttt{jleang@cs.washington.edu}
\And  % use '\and' if you need 'another row' of author names
% 4th. author
 Bill Howe\\
      University of Washington\\
      Seattle, Washington\\
      \texttt{billhowe@cs.washington.edu}
% 5th. author
\And Dan Suciu\\
      University of Washington\\
      Seattle, Washington\\
      \texttt{suciu@cs.washington.edu}
}

\maketitle

\begin{abstract}
  Machine learning algorithms are commonly specified in linear algebra
  (LA). LA expressions can be rewritten into more efficient forms, by
  taking advantage of input properties such as \textit{sparsity}, as
  well as program properties such as \textit{common subexpressions}
  and \textit{fusible operators}.  The complex interaction among these
  properties' impact on the execution cost poses a challenge to
  optimizing compilers. Existing compilers resort to intricate
  heuristics that complicate the codebase and add maintenance cost
  but fail to search through the large space of equivalent LA
  expressions to find the cheapest one.  We introduce a general
  optimization technique for LA expressions, by converting the LA
  expressions into Relational Algebra (RA) expressions, optimizing the
  latter, then converting the result back to (optimized) LA
  expressions.  One major advantage of this method is that it is
  complete, meaning that any equivalent LA expression can be found
  using the equivalence rules in RA.  The challenge is the major size
  of the search space, and we address this by adopting and extending a
  technique used in compilers, called \textit{equality saturation}. We
  integrate the optimizer into SystemML and validate it empirically
  across a spectrum of machine learning tasks; we show that we can
  derive all existing hand-coded optimizations in SystemML, and
  perform new optimizations that lead to speedups from $1.2X$ to $5X$.
%% 
  %% a technique for optimizing LA expressions with sum, product and
  %% aggregate operations, while exploiting sparsity and common
  %% subexpressions, and supporting operator fusion and custom
  %% functions. Our optimizer leverages relational algebra as an
  %% intermediate representation to completely represent the search
  %% space; it then searches for the optimal expression with a
  %% holistic optimization technique called \textit{equality
  %% saturation}. We integrate the optimizer into SystemML, and run
  %% experiments across a spectrum of machine learning tasks. The
  %% experiments show our approach can derive all optimizations found
  %% by hand-coded heuristics in SystemML, and also performs new
  %% optimizations that lead to speedups from $1.2X$ to $5X$.
\end{abstract}

\section{Introduction}
Consider the Linear Algebra (LA) expression $sum((X-UV^T)^2)$ which
defines a typical loss function for approximating a matrix $X$ with a
low-rank matrix $UV^T$. Here, $sum()$ computes the sum of all matrix
entries in its argument, and $A^2$ squares the matrix $A$
element-wise. Suppose $X$ is a sparse, 1M x 500k matrix, and suppose
$U$ and $V$ are dense vectors of dimensions 1M and 500k respectively.
Thus, $UV^T$ is a rank 1 matrix of size 1M x 500k, and computing it
naively requires 0.5 trillion multiplications, plus memory allocation.
Fortunately, the expression is equivalent to
$sum(X^2) - 2U^TXV + U^TU * V^TV$.  Here $U^TXV$ is a scalar that can
be computed efficiently by taking advantage of the sparsity of $X$,
and, similarly, $U^TU$ and $V^TV$ are scalar values requiring only 1M
and 500k multiplications respectively.

Optimization opportunities like this are ubiquitous in machine
learning programs. State-of-the-art optimizing compilers such as
SystemML~\cite{DBLP:reference/bdt/Boehm19},
OptiML\cite{DBLP:conf/icml/SujeethLBRCWAOO11}, and
Cumulon\cite{DBLP:conf/sigmod/HuangB013} commonly implement syntactic
rewrite rules that exploit the algebraic properties of the LA
expressions. For example, SystemML includes a rule that rewrites the
preceding example to a specialized operator \footnote{See the
  \href{https://systemml.apache.org/docs/0.12.0/engine-dev-guide.html}{SystemML
    Engine Developer Guide} for details on the weighted-square loss
  operator \texttt{wsloss}. } to compute the result in a streaming
fashion.  However, such syntactic rules fail on the simplest
variations, for example SystemML fails to optimize $sum((X+UV^T)^2)$,
where we just replaced $-$ with $+$.  Moreover, rules may interact
with each other in complex ways.  In addition, complex ML programs
often have many common subexpressions (CSE), that further interact
with syntactic rules, for example the same expression $UV^T$ may occur
in multiple contexts, each requiring different optimization rules.

%%%% (Dan: I found the next paragraph a bit too vague, so I shorted
%%%% the idea)
% Moreover, these rules resort to heuristics when
% multiple factors impact execution cost. The interaction among these
% factors can become complex: for example, if $UV^T$ in our example
% appears somewhere else in the same program, the original form may be
% preferable since all of its parents can share the computation result,
% and common-subexpression (CSE) elimination avoids recomputing the
% product. The compiler engineer may then add heuristics to only perform
% the rewrite when $UV^T$ does not have other parents. However, if its
% other parents can also be optimized by other rewrites, the heuristics
% will lead the compiler to miss the optimization opportunity. In
% general, heuristics like this complicate the codebase, add maintenance
% cost, and may miss high-impact optimizations.

% These opportunities and challenges call for a more principled approach
% to optimizing linear algebra expressions in machine learning programs.
In this paper we describe \sys, a novel optimization approach for
complex linear algebra programs that leverages relational algebra as
an intermediate representation to completely represent the search
space. SPORES first transforms LA expressions into traditional
Relational Algebra (RA) expressions consisting of joins $*$, union $+$
and aggregates $\sum$.  It then performs a cost-based optimizations on
the resulting Relational Algebra expressions, using only standard
identities in RA.  Finally, the resulting RA expression is converted
back to LA, and executed.

A major advantage of \sys\ is that the optimization rules in RA are
complete.  Linear Algebra seems to require an endless supply of clever
rewrite rules, but, in contrast, by converting to RA, we can prove
that SPORES is complete.  The RA expressions in this paper are
over $K$-relations~\cite{DBLP:conf/pods/GreenKT07}; a tuple $X(i,j)$
is no longer true or false, but has a numerical value, e.g. $5$; in
other words, the RA expressions that result from LA expressions are
interpreted over bags instead of sets.  A folklore theorem states that
two Unions of Conjunctive Queries over bag semantics are equivalent
iff they are isomorphic\footnote{This was claimed, for conjunctive
  queries only, in Theorem 5.2 in~\cite{DBLP:conf/pods/ChaudhuriV93}
  but a proof was never produced; a proof was given for bag-set
  semantics in~\cite{DBLP:journals/tocl/CohenSN05}.  See the
  discussion in~\cite{DBLP:conf/icdt/Green09}.}, which implies that
checking equivalence is decidable.  (In contrast, {\em containment} of
two UCQs with bag semantics is
undecidable~\cite{DBLP:journals/tods/IoannidisR95}; we do not consider
containment in this paper.)  We prove that our optimizer rules are
sufficient to convert any RA expression into its {\em canonical form},
i.e.  to an UCQ, and thus can, in principle, can discover all
equivalent rewritings.

However, we faced a major challenge in trying to exploit the
completeness of the optimizer.  The search space is very large,
typically larger than that encountered in standard database
optimizers, because of the prevalence of unions $+$, large number of
aggregates $\sum$, and frequent common subexpressions. To tackle this,
\sys\ adopts and extends a technique from compilers called {\em
  equality saturation}~\cite{DBLP:journals/corr/abs-1012-1802}.
% to search the space.  Equality saturation 
It uses a data structure called the E-Graph \cite{10.5555/909447} to
compactly represent the space of equivalent expressions, and equality
rules to populate the E-Graph, then leverages constraint solvers to
extract the optimal expression from the E-Graph.  

We have integrated \sys\ into
SystemML~\cite{DBLP:reference/bdt/Boehm19}, and show that it can
derive all hand-coded rules of SystemML.  We evaluated \sys\ on a
spectrum of machine learning tasks, showing competitive performance
improvement compared with more mature heuristic-based optimizers. Our
optimizer rediscovers all optimizations by the latter, and also finds
new optimizations that contribute to speedups of 1.2X to 5X.

We make the following contributions in this paper:
\begin{enumerate}
\item We describe a novel approach for optimizing complex Linear
  Algebra expressions by converting them to Relational Algebra, and
  prove that this approach is complete (Sec.~\ref{sec:la:to:ra}).
\item We present search algorithm based on Equality Saturation that
  can explore a large search space while reusing memory (Sec.~\ref{explore}).
\item We conduct an empirical evaluation of the optimizer using
  several real-world machine learning tasks, and demonstrate it's
  superiority over an heuristics-driven optimizer in SystemML
  (Sec.~\ref{sec:evaluation}).
\end{enumerate}{}

% \dan{Add: ``We make the following contributions in this paper:''
%   Bullet 1, bullet 2, etc}
% \remy{Added above.}

\section{Representing the Search Space} \label{sec:la:to:ra}

\subsection{Rules \texorpdfstring{$R_{LR}$}{RLR}: from  LA to RA and Back}
In this section we describe our approach of optimizing LA expressions
by converting them to RA.  The rules converting from LA to RA and back
are denoted $R_{LR}$.

To justify our approach, let us revisit our example loss function
written in LA and attempt to optimize it using standard LA identities.
Here we focus on algebraic rewrites and put aside concerns about the
cost model. Using the usual identities on linear algebra expressions,
one may attempt to rewrite the original expression as follows:
\begin{align*}
& sum((X-UV^T)^2) \\
= &sum((X-UV^T)*(X-UV^T)) \\
=&sum(X^2-2X*UV^T+(UV^T)^2) \\
= & sum(X^2) - 2sum(X*UV^T) + sum((UV^T)^2)
\end{align*}
% \dan{$(UV^T)^2$ should be $sum(UV^T)^2$.}
At this point we are stuck trying to rewrite $sum(X*UV^T)$ (recall
that $*$ is element-wise multiplication); it turns out to be equal to
$sum(U^TXV)$, for any matrices $X,U,V$ (and it is equal to the scalar
$U^TXV$ when $U, V$ are column vectors), but this does not seem to
follow from standard LA identities like associativity, commutativity,
and distributivity.  Similarly, we are stuck trying to rewrite
$sum((UV^T)^2)$ to $sum(U^TU * V^TV)$.  Current systems manually add
syntactic rewrite rules, whenever such a special case is deemed
frequent enough to justify extending the optimizer.

\begin{figure}
\centering
%%% (Dan: I replaced the multiplicies 1,2 with different ones so
%%% readers don't confuse them with indices in the matrix)
\begin{tabular}{ccccc}
&
     $A$ & $x$ & $A * x^T$ & $Ax$\\
     \hline
     \\[-1em]
LA: &
     $\begin{bmatrix}
    0 & 5 \\
    7 & 0
    \end{bmatrix}$ & $\begin{bmatrix}
    3\\
    2
    \end{bmatrix}$ & $\begin{bmatrix}
    0 & 10 \\
    21 & 0
    \end{bmatrix}$ & $\begin{bmatrix}
    10 \\
    21
    \end{bmatrix}$ \\
    \\[-0.8em]
    \hline
RA: &
    \begin{tabular}{|c|c|c}
    $i$ &$j$&\# \\
    \hline
        $1$ & $2$ & 5 \\
        $2$ & $1$ & 7
    \end{tabular}{} &
    \begin{tabular}{|c|c}
         $j$ &\# \\
         \hline
         1 & 3\\
         2 & 2
    \end{tabular}{} &
    \begin{tabular}{|c|c|c}
    $i$ &$j$ & \# \\
    \hline
        $1$ & $2$ & 10 \\
        $2$ & $1$ & 21
    \end{tabular}{} &
    \begin{tabular}{|c|c}
         $i$ & \# \\
         \hline
         1 & 10 \\
         2 & 21
    \end{tabular}{} 
    
\end{tabular}{}
    
\caption{Top: Linear Algebra manipulates matrices and vectors.
  $A*x^T$ is element-wise multiplication, computing the matrix
  $(A_{ij}x_j)_{ij}$, while $Ax$ is standard matrix-vector
  multiplication.  Bottom: their translations into $K$-relations and
  Relational Algebra operations.  Each relation is a bag, where the
  attribute $\#$ represents the multiplicity, i.e. $A(1,2)$ has
  multiplicity $5$ etc.  Then $A*x^T$ becomes a standard query
  $Q(i,j) = A(i,j)*x(j)$, while $Ax$ is a query with a group-by and
  aggregate, $Q(i) = \sum_j A(i,j)*x(j)$.}
    \label{matrixkrel}
\end{figure}{}

Instead, our approach is to expand out the LA expression element-wise.
For example, assuming for simplicity that $U,V$ are column vectors, we
obtain
\begin{align*}
& sum((UV^T)^2) = \textstyle{\sum_{i,j}} (U_i * V_j) * (U_i * V_j) \\
& = \textstyle{\sum_{i,j}}(U_i * U_i) * (V_j * V_J) \\
& = ( \textstyle{\sum_i} U_i * U_i) * ( \textstyle{\sum_j} V_j * V_j)  \\
& = U^TU * V^TV
\end{align*}
The expressions using indices represent Relational Algebra
expressions.  More precisely, we interpret every vector, or matrix, or
tensor, as a $K$-relation~\cite{DBLP:conf/pods/GreenKT07} over the
reals. In other words we view $X_{ij}$ is a tuple $X(i,j)$ whose
``multiplicity'' is the real value of that matrix element.  We
interpret point-wise multiply as natural join; addition as union; sum
as aggregate; and matrix multiply as aggregate over a join\footnote{In
  the implementation, we use outer join for point-wise multiply and
  addition, where we multiply and add the matrix entries
  accordingly. In this paper we use join and union to simplify
  presentation.}.
Figure~\ref{matrixkrel} illustrates the correspondence between LA and
RA.  We treat each matrix entry $A_{ij}$ as the multiplicity of tuple
$(i, j)$ in relation $A$ under bag semantics. For example
$A_{2,1} = 7$, therefore the tuple $(2,1)$ has multiplicity of $7$ in
the corresponding relation.  $A*x^T$ denotes element-wise
multiplication, where each element $A_{ij}$ of the matrix is
multiplied with the element $x_j$ of the row-vector $x^T$.  In RA it
is naturally interpreted as the natural join $A(i,j) \Join x(j)$,
which we write as $A(i,j)*x(j)$.  Similarly, $Ax$ is the standard
matrix-vector multiplication in LA, while in RA it becomes a query
with a group by and aggregate, which we write as $\sum_j A(i,j)*x(j)$.
Our $K$-relations are more general than bags, since the entry of a
matrix can be a real number, or a negative number; they correspond to
$K$-relations over the semiring of reals $(\R, 0, 1, +, *)$.

% Of course, a matrix can have an entry that is not an integer, which
% would not make sense as multiplicity. In general, the matrix entries
% corresponds to the semiring of reals which serve as relation
% annotations \cite{DBLP:conf/pods/GreenKT07} in the $K$-relation.
\begin{table}
\centering
\begin{tabular}{llll}
&name & type & syntax \\ \hline \multirow{7}{*}{\rotatebox[origin=c]{90}{LA}}
  &mmult & $M_{M,L} \times M_{L,N} \rightarrow M_{M,N}$ & $AB$ or MxM
  \\ &elemmult & $M_{M,N} \times M_{M,N} \rightarrow M_{M,N}$ & $A*B$
  \\ &elemplus & $M_{M,N} \times M_{M,N} \rightarrow M_{M,N}$ & $A+B$ \\ &rowagg
  & $M_{M,N} \rightarrow M_{M,1}$ & $sum_{row} A$ \\ &colagg & $M_{M,N}
  \rightarrow M_{1,N}$ & $sum_{col} A$ \\ &agg & $M_{M,N} \rightarrow M_{1,1}$ &
  $sum A$ \\
%&power & $M  \times Int \rightarrow M$ & $A^k$ \\
&transpose & $M_{M,N} \rightarrow M_{N,M}$ & $A^\tr$ \\ \hline
  \multirow{2}{*}{\rotatebox[origin=c]{90}{conv.}} &bind & $M_{M,N} \times [i,j]
  \rightarrow R_{i:M,j:N}$ & $\bind{i,j} A$ \\ &unbind & $R_{i:M,j:N} \times
              [i,j] \rightarrow M_{M,N}$ & $\unbind{i}{j} A$ \\ \hline
              \multirow{2}{*}{\rotatebox[origin=c]{90}{RA}} &join & $R_{S_1}
              \times R_{S_2} \rightarrow R_{S_1 \cup S_2}$ & $A \rmul B$
              \\ &union & $R_{S_1} \times R_{S_2} \rightarrow R_{S_1 \cup S_2}$
              & $A \radd B$ \\ &agg & $R_S \times U \rightarrow R_{S \setminus
                U}$ & $\sum_{U} A$ \\
\end{tabular}
\caption{LA and RA Operators. The type $M_{M,N}$ is a matrix of size $M \times
  N$; $[i,j]$ is a list of attribute names; $R_{i:M,j:N}$ is a relation with
  attributes $i$ of size $M$ and $j$ of size $N$; $S_1, S_2, S,$ and $U$ are
  sets of attributes.}
\label{tPlanOps}
\end{table}
\begin{figure}
\centering
\begin{minipage}{0.6\textwidth}
\begin{enumerate}
  \item $A*B \rightarrow \unbind{i}{j}(\bind{i,j}A \rmul \bind{i,j}B)$.
  \item $A+B \rightarrow \unbind{i}{j}(\bind{i,j}A \radd \bind{i,j}B)$.
  \item $sum_{row} A \rightarrow [-i,\_ ] \sum_j \bind{i,j} A$. Similar for
    $sum_{col}$, $sum$.
  \item $AB \rightarrow \unbind{i}{k} \sum_j (\bind{i,j}A \rmul \bind{j,k}B)$.
  \item $A^\tr \rightarrow \unbind{j}{i} \bind{i,j}A$.
  \item $A - B \rightarrow A + (-1) * B$
\end{enumerate}
\end{minipage}{}
\caption{LA-to-RA Ruleset $R_{LR}$. Notice that $*$ means point-wise multiply for matrices,
  but means natural join for relations. Similarly, $+$ means addition for matrices, but
  union for relations.}
\label{RMR}
\end{figure}

% Following our relational semantics, $sum((X-UV^T)^2)$ is interpreted
% as
% $\sum_{ij} (X_{ij} - U_{i} \bowtie V_{j}) \bowtie (X_{ij} - U_{i}
% \bowtie V_{j})$. Note that because $U$ and $V$ are vectors, the
% aggregate in their matrix product $\sum_{\varnothing} U_i \bowtie V_j$
% is over an empty set of attributes and therefore eliminated.
% Hereafter we will overload $*$ for join and $+$ for union to keep our
% formulas intuitive.  Table~\ref{tPlanOps} provides the list of LA and
% RA operators.  

We now describe the general approach in \sys.  The formal definition
of LA and RA are in Table~\ref{tPlanOps}.  LA consists of seven
operators, which are those supported in
SystemML~\cite{DBLP:reference/bdt/Boehm19}.  RA consists of only three
operators: $*$ (natural join), $+$ (union), and $\sum$ (group-by
aggregate).  Difference is represented as $A-B = A + (-1)B$ (this is
difference in $\R$; we do not support bag difference, i.e.  difference
in $\N$ like $3-5=0$, because there is no corresponding operation in
LA), while selection can be encoded by multiplication with relations
with 0/1 entries.  We call an expression using these three RA
operators an {\em RPlan}, for Relational Plan, and use the terms RPlan
and RA/relational algebra interchangeably.  Finally, there are two
operators, {\em bind} and {\em unbind} for converting between
matrices/vectors and $K$-relations.

The translation from LA to RA is achieved by a set of rules, denoted
$R_{LR}$, and shown in Figure~\ref{RMR}. The bind operator
$\bind{i,j}$ converts a matrix to a relation by giving attributes
$i,j$ to its two dimensions; the unbind operator $\unbind{i}{j}$
converts a relation back to a matrix. For example,
$\unbind{j}{i}\bind{i,j}A$ binds $A$'s row indices to $i$ and its
column indices to $j$, then unbinds them in the opposite order,
thereby transposing $A$.

\sys\ translates a complex LA expression into RA by first applying the
rules $R_{LR}$ in Figure~\ref{RMR} to each LA operator, replacing it
with an RA operator, preceded by {\em bind} and followed by {\em
  unbind}. Next, it eliminates consecutive unbind/bind operators,
possibly renaming attributes, e.g.  $\bind{k,l} \unbind{i}{j} A$
becomes $A[i \to k, j \to l]$, which indicates that the attributes $i$
and $j$ in $A$'s schema should be renamed to $k$ and $l$, by
propagating the rename downward into $A$.  As a result, the entire LA
expression becomes an RA expression (RPlan), with {\em bind} operators
on the leaves, and {\em unbind} at the top.
% The result of eliminating unbind and bind operators is an RA
% expression with a bind operator above every (nonscalar) input and an
% unbind operator below every (nonscalar) output.
For an illustration, the left DAG in Figure~\ref{fig:radags} shows the
expression $sum((X - UV^T)^2)$ translated to relational algebra.

\begin{figure}
\centering
\begin{minipage}{0.6\textwidth}
\begin{enumerate}
  \item\label{RRC_mp} $A \rmul (B \radd C) = A \rmul B \radd A \rmul C$
  \item\label{RRC_ap} $\sum_i (A \radd B) = \sum_i A \radd \sum_i B$
  \item\label{RRC_ma} If $i \not\in A$, $A \rmul \sum_i B = \sum_i (A \rmul B)
    \quad$ (else rename $i$)
  \item\label{RRC_aa} $\sum_i \sum_j A = \sum_{i,j} A$
%   \item $\sum_i A \rightarrow \sum_{i,j} A$
  \item\label{RRC_ac} If $i \not\in Attr(A)$, then $\sum_i A = A \rmul dim(i)$
  \item\label{RRC_pp} $A \radd (B \radd C) = +(A, B, C) \quad$ (assoc. \& comm.)
  \item\label{RRC_mm} $A \rmul (B \rmul C) = *(A, B, C) \quad$ (assoc. \& comm.)
\end{enumerate}
\end{minipage}
\caption{RA equality rules $R_{EQ}$. $*$ means natural join and $+$ means
  union.}
\label{RRC}
\end{figure}

\begin{figure}
    \centering
    \includegraphics[width=0.6\linewidth]{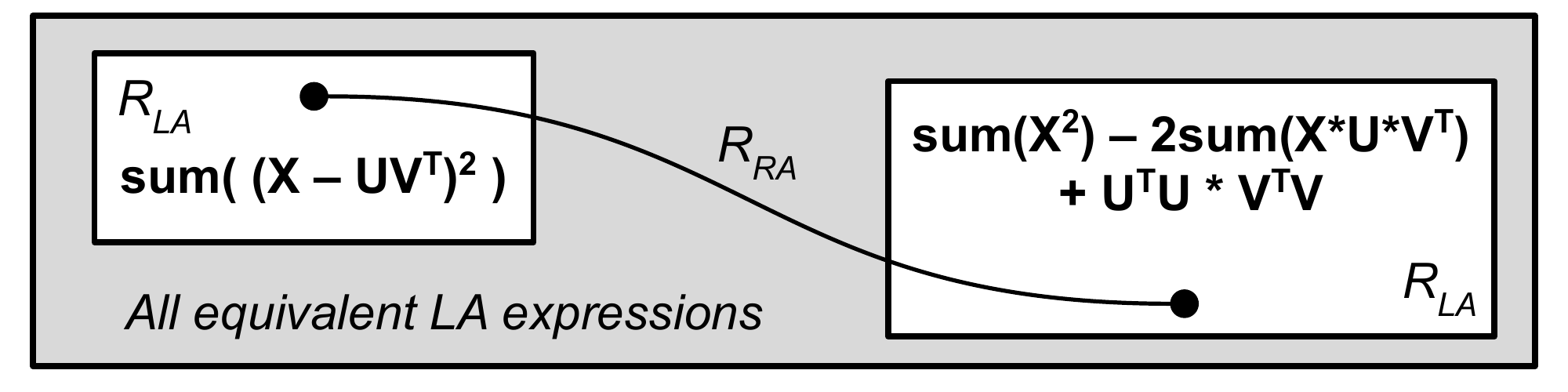}
    \caption{Venn diagram of LA and/or RA expressions.  Each of the two
      white islands represent LA expressions that can be proven
      equivalent using identities in LA; there is no way to move from
      one island to the other by using LA identities.  The large gray
      box shows the set of expressions that can be proven equivalent
      by using RA identities.  This allows us to connect the two
      islands. }
    \label{wormhole}
\end{figure}{}
\begin{figure}
    \centering
\begin{tikzcd}
e_{LA} \arrow[rr, "=" description, no head, dotted] \arrow[d, "R_{LR}"'] &           & e'_{LA} \arrow[d, "R_{LR}"]  \\
e_{RA} \arrow[r, "R_{EQ}"]                                               & \mathcal{C}(e_{RA}) \equiv \mathcal{C}(e'_{RA})  & e'_{RA} \arrow[l, "R_{EQ}"']
\end{tikzcd}
    
    \caption{Two LA expressions $e_{LA}$ and $e'_{LA}$ are equivalent \textit{iff} their
    relational counterparts $e_{RA}$ and $e'_{RA}$ have isomorphic
      canonical forms $\mathcal{C}(e_{RA}), \mathcal{C}(e'_{RA})$. $R_{LR}$ translates a LA expression to RA and $R_{EQ}$ are the relational equality rules. }
    \label{proofsketch}
\end{figure}
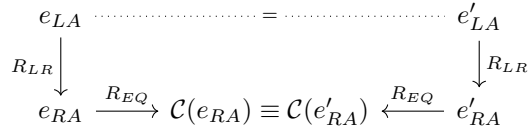{}

\begin{figure}
    \centering
    \includegraphics[width=0.6\linewidth]{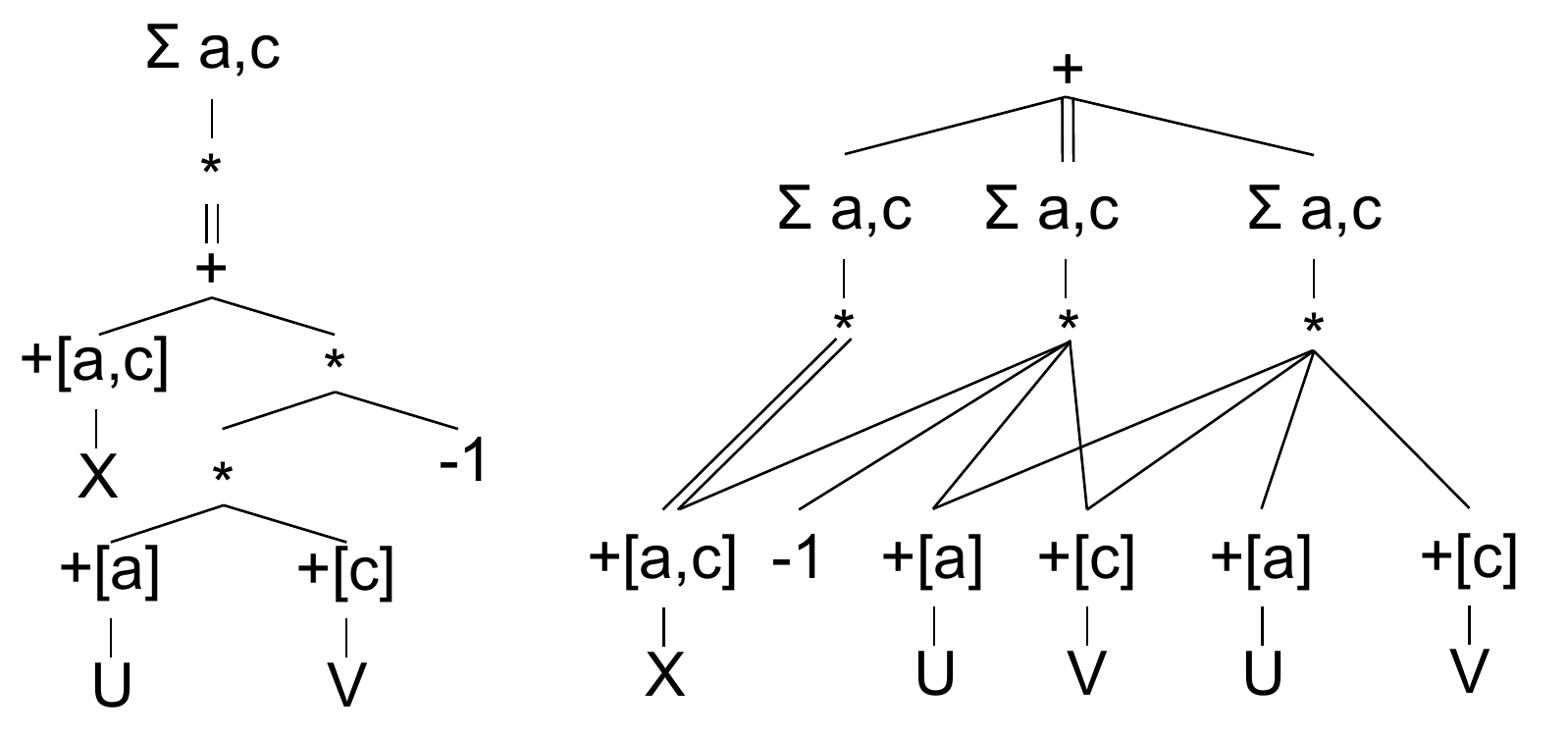}
    \caption{RA DAGs for $sum((X-UV^T)^2)$ (left) and its canonical form (right). }
    \label{fig:radags}
\end{figure}{}

% \begin{figure*}
%     \centering \includegraphics[width=\textwidth]{parrot}
%     \caption{Rewriting details using $R_{LR}$ and $R_{EQ}$ and the
%       optimization's impact on execution cost. Here we show the general case
%       where $U$ and $V$ are matrices, and $X*UV^T$ is computed with a fused
%       chain multiplication operator that exploits sparsity. Dashed lines
%       highlight the initial RPlan and its canonical form. }
%     \label{parrot}
% \end{figure*}

\subsection{Rules \texorpdfstring{$R_{EQ}$}{REQ} : from RA to RA}\label{sec:req}

The equational rules for RA consists of seven identities shown in
Figure~\ref{RRC}, and denoted by $R_{EQ}$.  The seven rules are
natural relational algebra identities, where $*$ corresponds to
natural join, $+$ to union (of relations with the same schema) and
$\sum_i$ to group-by and aggregate.  In rule~\ref{RRC_ac},
$i \notin Attr(A)$ means that $i$ is not an attribute of $A$, and
$dim(i)$ is the dimension of index $i$.  For a very simple
illustration of this rule, consider $\sum_i 5$.  Here $5$ is a
constant, i.e. a relation of zero arity, with no attributes.  The rule
rewrites it to $5 dim(i)$, where $dim(i)$ is a number representing the
dimension of $i$.

%%% (Dan: I think the example below is wrong. In RA you must have
%%% indices, can't write $x+1$.
%  $\sum_i (x + 1)$ where $x$ is a vector of length 10 and
% $x + 1$ adds $1$ with each entry of $x$. Following rule~\ref{RRC_ap},
% pushing down the summation rewrites the expression to
% $\sum_i x + \sum_i 1$, and now rule~\ref{RRC_ac} rewrites $\sum_i 1$
% to $|i| * 1 = 10$.

\subsection{Completeness of the Optimization Rules}\label{ranf}

As we have seen at the beginning of this section, when rewriting LA
expressions using identities in linear algebra we may get stuck.  This
is illustrated in Figure~\ref{wormhole}, which shows two white islands
representing sets of provably equivalent LA expressions; the simple
identities in LA allow us to move within one white island, but are
insufficient to move from one island to the other.  Instead, by
rewriting the expressions to RA, the seven identities in $R_{EQ}$ are
much more powerful, and allow us to prove them equivalent.  We prove
here that this approach is {\em complete}, meaning that, if two LA
expressions are semantically equivalent, then their equivalence can be
proven by using rules $R_{EQ}$.  The proof consists of two parts: (1)
the rules $R_{EQ}$ are sufficient to convert any RA expression $e$ to
its {\em normal form} (also called {\em canonical form})
$\mathcal{C}(e)$, and back, (2) two RA expressions $e, e'$ are
semantically equal iff they have isomorphic normal forms,
$\mathcal{C}(e) \equiv \mathcal{C}(e')$.
% 
% 
% 
% In general, any equivalent LA expressions can be rewritten to each other using
% our RA representation and equalities. In other words, the relational equalities
% $R_{EQ}$ are \textit{complete} w.r.t. linear algebra semantics. For
% optimization, this means RPlan and its equalities completely represent the
% search space of the optimization problem: we can rewrite an LA expression to any
% of its equals if we first translate it into RA, then apply a chain of equalities
% from $R_{EQ}$, and finally translate the result back to LA. As
% Figure~\ref{wormhole} illustrates, while the usual linear algebra identities
% $R_{LA}$ can only create disjoint equality space for our example expression pre-
% and post-optimization, the relational rules connects the space and rewrite the
% slow expression to the fast one. We sketch a proof of completeness in this
% section.
% 
Figure~\ref{proofsketch} shows the main steps at a high level: two LA
expressions are semantically equivalent if and only if their canonical
form in RA are isomorphic, where $R_{LR}$ translates each expression
to RA and $R_{EQ}$ takes the RA expression to its canonical form.
Throughout this section we write $e\ R_{EQ}^*\ e'$ when $e, e'$ can be
proven equal by using the identities in $R_{EQ}$ (Fig.~\ref{RRC}).

% By semantically equivalent, we
% mean two expressions evaluate to the same result given any same input (variable
% assignment). We define isomorphism of RA canonical forms in Section~\ref{ranf}.

{\bf Normal Form} The normal form of an RPlan is similar in spirit to
the canonical form of polynomials as a sum of monomials, except that
the monomial terms also include aggregations.  As for polynomials, we
combine equal factors by introducing power, e.g. $X * X$ becomes
$X^2$, and combine isomorphic monomials by introducing constant
coefficients, e.g. $3X^2Y+5X^2Y$ becomes $8X^2Y$.  For example, the
following RPlan denotes a 1-dimensional vector
$Q(i) = (\sum_j x(i,j) * (\sum_k y(j,k)*x(i,j)) + \sum_{m,n} x(i,m)^2
y(m,n)$ and its canonical form is $2 \sum_{j,k} x(i,j)^2y(j,k)$.
Formally:

% 
% 
% Polynomials in their
% canonical form may not be the most efficient form to execute, but they
% do act as a canonical form and also serve as a starting point for
% finding more efficient forms. The analogy to polynomials extends
% further: the canonical form of polynomials combines \textit{like
%   terms}, i.e., two terms that have the same variables. Instead of
% writing $3X^2Y + 5X^2Y$, for example, we write its canonical form as
% $8X^2Y$. For RA expressions, we determine whether two terms are like
% by whether there exists an isomorphism mapping one to the other by
% renaming attributes (up to a constant scalar).
% 
% The canonical form for RA expressions is the sum-of-aggregations-of-products
% with like terms combined, as follows:

\begin{definition}
  An RA expression is \textbf{canonical}, or in \textbf{normal form},
  if it is the sum (i.e. $+$) of $n$ {\em monomials}, where each
  monomial $i$ consists of a constant $c_i$ multiplied by an
  aggregation over a (possibly empty) set of attributes $A_i$ (for
  $1 \leq i \leq n$), of a product of $m_i$ {\em factors}, where each
  factor is some matrix or vector $x_{ij}$ (for $1 \leq i \leq n$,
  $1 \leq j \leq m_i$) indexed by some index variables (not shown),
  and possibly raised to some power:

\[c_1 \sum_{A_1} \lrparen{x_{11}^{k_{11}} * \dots * x_{1m_1}^{k_{1m_1}}} %+ \sum_{A_2} (b_{21} * \dots * b_{2m_2})
+ \dots + c_n \sum_{A_n} \lrparen{x_{n1}^{k_{n1}} * \dots * x_{nm_n}^{k_{nm_n}}}\]

We further assume that no monomial contains the same factor twice
(otherwise, replace it with a single factor with a higher power
$k_{ij}$) and no two monomials are isomorphic (otherwise we replace
them with a single monomial with a larger coefficient $c_i$)
% 
% Further, there must not exist, for any $1 \leq i < j \leq n$, an isomorphism
% $\phi$ renaming attributes in the schema of $b_{i1} * \dots * b_{im_i}$ to
% attributes in the schema of $b_{j1} * \dots * b_{jm_j}$ such that $m_i=m_j$,
% $\phi(A_i) = A_j$, and $\phi\{b_{i1} * \dots * b_{im_i}\} = \{b_{j1} * \dots *
% b_{jm_j}\}$.
% 
\end{definition}
The order of summands and multiplicands is insignificant because $*$
and $+$ are commutative and associative.  We notice that, unlike
traditional polynomials, here the same matrix name may occur in
different factors, for example in $\sum_{i,j,k} x(i,j)*x(j,k)*x(k,i)$
the matrix $x$ occurs three times, but the factors are different
i.e. we cannot write $x^3$.

The first of the completeness proof consists in showing that every
expression $e$ in RA is equivalent to some normal form
$\mathcal{C}(e)$, and, moreover, that their equivalence can be proven
using the rules $R_{EQ}$ in Figure~\ref{RRC}.

\begin{lemma}\label{lCanonPreservesSemantics}
$\forall e \in RA$ there exists a canonical form $\mathcal{C}(e)$ and,
moreover, their equivalence follows from the rules in $R_{EQ}$, in
notation: $e\ R_{EQ}^*\  \mathcal{C}(e)$.
\end{lemma}

The proof is a straightforward induction on the structure of $e$.  We
apply distributivity of $*$ over $+$, then pull out the summation
$\sum$; we omit the details.  We illustrate in Figure~\ref{fig:radags}
the canonical form of the expression $sum((X-UV^T)^2)$.  Notice that,
since the rules in $R_{EQ}$ are sound, it follows that any expression
$e$ has {\em the same semantics} as its canonical form
$\mathit{C}(e)$.

The second step of the completeness proof is to show that canonical
forms are unique up to isomorphism.

\begin{lemma} \label{lemma:unique:nf} (\textbf{Uniqueness of RA Normal
    Form}) Let $e_1, e_2$ be two RA expressions in normal form.
  Suppose that they have the same semantics, i.e. $e_1=e_2$ for all
  inputs with arbitrary dimensions.  Then their canonical forms are
  isomorphic.
%% 
%% Assuming matrix dimensions in any variable are sufficiently large, two RA
%% expressions are semantically equivalent if and only if they have isomorphic
%% canonical forms:
%% 
%% \[\forall p_1, p_2 \in \RPLAN{},\, \mathcal{C}(p_1) \equiv \mathcal{C}(p_2) \iff
%% \forall d,\, p_1(d) = p_2(d)\]
\end{lemma}

We give a proof of this lemma in the appendix.
Here, we comment on a subtle issue, namely the requirement that
$e_1=e_2$ on inputs of {\em any} dimensions is necessary.  For
example, if we restrict the matrices $x, y$ to be of dimension
$1 \times 1$, then the expressions $\sum_{i,j} x(i,j)*y(i,j)$ and
$\sum_{i,j} x(i,j)*y(j,i)$ have the same semantics, but different
canonical form.  For another example, consider three vectors $x,y,z$
of dimension $N$.  Then, if $N \leq 2$ these two expressions are
identical: $\sum_{i,j,k} x(i)*y(j)*z(k)+2\sum_i x(i)*y(i)*z(i)$ and
$\sum_{i,j}x(i)*y(i)*z(j)+\sum_{i,j}x(i)*y(j)*z(i)+\sum_{i,j}x(j)*y(i)*z(i)$,
although they are not equal in general.

We are now ready to establish the completeness of RA equalities, by
showing any equivalent LA expressions can be rewritten to each other
through the translation rules $R_{LR}$ and the canonicalization rules
$R_{EQ}$:

\begin{theorem} (\textbf{Completeness of $R_{EQ}$}) Two LA expressions
  are semantically equivalent if and only if their relational form is in the
  transitive closure of $R_{EQ}$ rules:
\[\forall e_1, e_2 \in \MPLAN{},\, \forall d. e_1(d) = e_2(d) 
\iff R_{LR}(e_1) R_{EQ}^* R_{LR}(e_2)\]
\end{theorem}
Here $R_{LR}(e)$ translates LA expression $e$ into RA.

\begin{proof}
  Translating $e_1$ and $e_2$ to RA preserves semantics under
  $R_{LR}$. By Lemma~\ref{lCanonPreservesSemantics} normalizing
  $R_{LR}(e_1)$ and $R_{LR}(e_2)$ preserves semantics. By the
  uniqueness of the normal form (Lemma~\ref{lemma:unique:nf})
  \[\mathcal{C}(R_{LR}(e_1)) \equiv \mathcal{C}(R_{LR}(e_2)) \iff
    R_{LR}(e_1) = R_{LR}(e_2)\] Since every rule in $R_{EQ}$ is
  reversible,
  \[R_{LR}(e_1) = R_{LR}(e_2) \iff R_{LR}(e_1) R_{EQ}^* R_{LR}(e_2)\]
\end{proof}

\section{Exploring the Search Space} \label{explore}

With a complete representation of the search space by relational algebra, our
next step is to explore this space and find the optimal expression in it.
Traditional optimizing compilers commonly resort to heuristics to select from
available rewrites to apply. SystemML implements a number of heuristics for its
algebraic rewrite rules, and we discuss a few categories of them here.

\textsc{Competing or Conflicting Rewrites} The same expression may be eligible
for more than one rewrites. For example, $sum(AB)$ rewrites to
$sum(sum_{col}(A)^T*sum_{row}(B))$, but when both $A$ and $B$ are vectors the
expression can also be rewritten to a single dot product. SystemML then
implements heuristics to only perform the first rewrite when the expression is
not a dot product. In the worst case, a set of rules interacting with each other
may create a quadratic number of such conflicts, complicating the codebase.

\textsc{Order of Rewrites} Some rewrite should be applied after others to be
effective. For example, $X/y$ could be rewritten to $X*1/y$ which may be more
efficient, since SystemML provides efficient implementation for sparse
multiplication but not for division. This rewrite should occur before constant
folding; otherwise it may create spurious expressions like $X / (1/y)
\rightarrow X * (1/(1/y))$, and without constant folding the double division
will persist. However, a rewrite like $1/(1 + exp(-X)) \rightarrow sigmoid(X)$
should come after constant folding, in order to cover expressions like $(3-2)/(1
+ exp(-X))$. Since SystemML requires all rewrites to happen in one phase and
constant folding another, it has to leave out\footnote{Another reason to leave
  out this rewrite is that $X*1/y$ rounds twice, whereas $X/y$ only rounds
  once.} rewrites like $X/y \rightarrow X*1/y$.

\textsc{Dependency on Input / Program Properties} Our example optimization from 
$sum((X-UV^T)^2)$  to $sum(X^2) -2U^TXV + U^TU*V^TV$
improves performance only if $X$ is sparse. Otherwise,
computing $X^2$ and $X*UV^T$ would both create dense intermediates. Similarly,
some rewrites depend on program properties like common subexpressions. Usually,
these rewrites only apply when the matched expression shares no CSE with others
in order to leverage common subexpression elimination. Testing input and program
properties like this becomes boilerplate code, making implementation tedious and
adds burden to maintenance.

\textsc{Composing Rewrites} Even more relevant to us is the problem of composing
larger rewrites out of smaller ones. Our equality rules $R_{EQ}$ are very
fine-grained, and any rule is unlikely to improve performance on its own. 
Our example optimization from $sum((X-UV^T)^2)$ to $sum(X^2) - 2U^TXV + U^TU * V^TV$ 
takes around 10 applications of $R_{EQ}$ rules. 
 If an optimizer applies rewrites one by one, it is
then very difficult, if not impossible, for it to discover the correct sequence
of rewrites that compose together and lead to the best performance.

Stepping back, the challenge of orchestrating rewrites is known as the
\emph{phase-ordering problem} in compiler optimization. Tate
et al. \cite{DBLP:journals/corr/abs-1012-1802} proposed a solution to this
problem dubbed
\textit{equality saturation}, which we adapt and extend in SPORES.

\subsection{Equality Saturation}

Equality saturation optimizes an expression in two steps:

\textit{Saturation}: given the input expression, the optimizer  enumerates
equivalent expressions and collects them into a compact representation called
the E-Graph \cite{10.5555/909447}.

\textit{Extraction}: given a cost function, the optimizer selects the optimal
expression from the E-Graph. An expression is represented by a subgraph of the
E-Graph, and the optimizer uses a constraint solver to find the subgraph that is
equivalent to the input, and is optimal according to the cost function.

\subsubsection*{The E-Graph Data Structure}

\begin{figure}
\centering
    \includegraphics[scale=0.6]{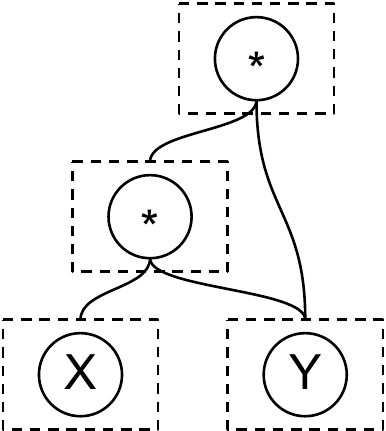} \qquad \includegraphics[scale=0.6]{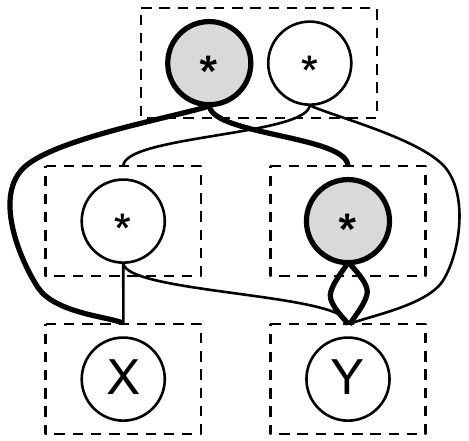}
    \caption{E-Graph representing $(X*Y)*Y$ (left), and the graph after applying
      associativity to the root (right). New nodes are in gray. Each dashed box is an E-Class.}
    \label{assoc}
\end{figure}

An E-Graph represents sets of equivalent expressions. A node in the graph is
called an E-Class, which contains the root operators of a set of equivalent
expressions. The edges are similar to the edges in an abstract syntax tree; but
instead of pointing from an operator directly to a child, each edge points from
an operator to an E-Class of expressions. For example, in Figure~\ref{assoc} the
top class in the middle represents the set of equivalent expressions \{$(X*Y)*Y,
X*(Y*Y)$\}. Note that the class represents two expressions, each with 2
appearances of $Y$ and one appearance of $X$, whereas each variable only appears
once in the E-Graph. This is because the E-Graph makes sure its expressions share
all possible common subexpressions. As the size of the graph grows, this
compression becomes more and more notable; in some cases a graph can 
represent a number of expressions exponential to its
size \cite{DBLP:journals/corr/abs-1012-1802}. We take advantage of this
compression in SPORES to efficiently cover vast portions of the search
space. If saturation, as described below, carries out to convergence, the E-Graph represents the search space exhaustively. 

An E-Graph can also be seen as an AND-OR DAG over expressions. Each E-Class is an OR node whose children are equivalent expressions from
which the optimizer chooses from. Each operator is an AND node whose children must all be picked if the operator itself is picked. In this paper we favor the terms E-Graph and E-Class to emphasize each OR node is an equivalence class. 

\subsubsection*{Saturating the E-Graph}
At the beginning of the optimization process, the optimizer instantiates the
graph by inserting the nodes in the syntax tree of the input expression one by
one in post order. For example, for input $(X*Y)*Y$, we construct the left graph
in Figure~\ref{assoc} bottom-up. By inserting in post order, we readily exploit
existing common subexpressions in the input. Once the entire input expression is
inserted, the optimizer starts to extend the graph with new expressions equivalent to
the input. It considers a list of equations, and matches either side of the
equation to subgraphs of the E-Graph. If an equation matches, the optimizer then
inserts the expression on the other side of the equation to the graph. For
example, applying the associative rule extends the left graph in
Figure~\ref{assoc} with $X*(Y*Y)$, resulting in the right graph.
Figure~\ref{eqsat} shows the pseudo code for this process. While inserting new
expressions, the optimizer checks if any subexpression of the new expression is
already in the graph. If so, it reuses the existing node, thereby exploiting all
possible common-subexpressions to keep the E-Graph compact. In
Figure~\ref{assoc}, only two $*$ are added since the variables $X$ and $Y$ are
already in the graph. Once the entire new expression has been added, the
optimizer then merges the newly created E-Class at its root with the E-Class
containing the matched expression, asserting them equal. Importantly, the
optimizer also propagates the congruent closure of this new equality. For
example, when $A+A$ is merged with $2*A$, the optimizer also merges $(A+A)^2$
with $(2*A)^2$. Figure~\ref{fig:egraphadd} shows the pseudo code for adding an
expression to E-Graph. This process of match-and-insert is repeated until the
graph stops changing, or reaching a user-specified bound on the number of 
saturation iterations. If this process does converge, that means no rule can add new
expressions to the graph any more. If the set of rules are complete, as is our
$R_{EQ}$, convergence of saturation implies the resulting E-Graph represents
the transitive closure of the equality rules applied to the initial expression.
In other words, it contains \textit{all} expressions equivalent to the input under the equality rules.

\begin{figure}
\centering
\begin{minipage}{0.5\textwidth}
\begin{lstlisting}[language=Python]
def saturate(egraph, equations):
  for eq in equations: 
    matches = egraph.match(eq.lhs)
    for eclass in matches: 
      ec = egraph.add(eq.rhs)
      egraph.merge(eclass, c)
\end{lstlisting}
\end{minipage}
    \caption{Pseudo code for saturating the E-Graph given a set of equalities.
      \texttt{match} returns the IDs of the root class of any matching subgraph;
      \texttt{merge} combines two E-Classes given their IDs; it also propagates
      the congruent closure of the new equality.
      Figure~\ref{fig:egraphadd} defines \texttt{add}.}
    \label{eqsat}
\end{figure}
\begin{figure}
\centering
\begin{minipage}{0.5\textwidth}
\begin{lstlisting}[language=Python]
def add(expr):
  ID = egraph.find(expr)
  if ID != NULL:
    return ID
  else:
    cids = expr.children.map(add)
    ID = egraph.insert(expr.op, cids)
    return ID
\end{lstlisting}
\end{minipage}
    \caption{Pseudo code for adding an expression to the E-Graph. \texttt{find}
      looks for the given expression in the E-Graph, and returns its root ID if
      it already exists, or \texttt{NULL} otherwise. \texttt{insert} adds the
      given operator to the E-Graph, and points its children to E-Classes with the
      given class IDs. }
    \label{fig:egraphadd}
\end{figure}

The outer loop that matches equations to the graph can be implemented by a more
efficient algorithm like the Rete algorithm \cite{DBLP:journals/ai/Forgy82} when
the number of equations is large. However, we did not find matching to be
expensive and simply match by traversing the graph. Our implementation uses the
E-Graph data structure from the \texttt{egg}~\cite{EGG} library. 

\subsubsection*{Dealing with Expansive Rules}
While in theory equality saturation will converge with well-constructed rewrite
rules, in practice the E-Graph may explode for certain inputs under certain rules. For example, a long
chain of multiplication can be rewritten to an exponential number of
permutations under associativity and commutativity (AC rules). If we 
apply AC rules everywhere applicable in each iteration, the graph would soon use
up available memory. We call this application strategy the \emph{depth-first}
strategy because it eagerly applies expansive rules like AC. 
AC rules by themselves rarely affect performance, and
SystemML also provides the fused \texttt{mmchain} operator that efficiently computes
multiplication chains, so permuting a chain is likely futile. In practice, AC
rules are useful because they can enable other rewrites. Suppose we have a rule
$R_{factor}: A*X + B*X \rightarrow (A+B)*X$ and an expression $U*Y + Y*V$.
Applying commutativity to $Y*V$ would then transform the expression to be
eligible for $R_{factor}$. With this insight, we change each saturation
iteration to sample a limited number of matches to apply per rule, instead of
applying all matches. This amounts to adding \texttt{matches = sample(matches, limit)} between line 3 and line 4 in Figure~\ref{fig:egraphadd}. 
Sampling encourages each rule to be considered equally often
and prevents any single rule from exploding the graph. This helps ensure good
exploration of the search space when exhaustive search is impractical. 
But when it is
possible for saturation to converge and be exhaustive, it still converges with high probability
when we sample matches. Our experiments in Section~\ref{overhead} show sampling
always preserve convergence in practice.

\subsubsection*{Extracting the Optimal Plan}

\begin{figure}
\centering
    \includegraphics[scale=0.6]{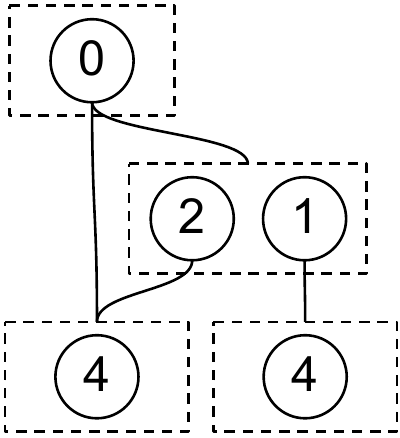}
    \caption{The CSE problem. Each node shows its cost. A greedy optimizer picks nodes with costs 0, 1 and both nodes with cost 4; but the optimal plan uses nodes with costs 0, 2 and shares the same node with cost 4.}
    \label{cse}
\end{figure}

A greedy strategy to extract the best plan from the saturated E-Graph is to
traverse the graph bottom-up, picking the best plan at each level. This assumes
the best plan for any expression also contains the best plan for any of its
subexpressions. However, the presence of common subexpressions breaks this
assumption. In Figure~\ref{cse} each operator node is annotated with its cost. 
Between the nodes with costs 1 and 2, a greedy strategy would choose 1, which
incurs total cost of $1+4=5$. The greedy strategy then needs to pick the root
node with cost 0 and the other node with cost 4, incurring a total cost of 9. 
However, the optimal strategy is to pick the nodes with 0, 2 and share the same
node with cost 4, incurring a total cost of 6. 

We handle the complexity of the search problem with a constraint
solver. We assign a variable to each operator and each E-Class, then construct
constraints over the variables for the solver to select operators that make up a
valid expression. The solver will then optimize a cost function defined over the
variables; the solution then corresponds to the optimal expression equivalent to
the input.

\subsubsection*{Constraint Solving and Cost Function}
We encode the problem of extracting the cheapest plan from the E-Graph with
integer linear programming (ILP). Figure~\ref{ilp} shows this encoding. For each
operator in the graph, we generate a boolean variable $B_{op}$; for each E-Class
we generate a variable $B_c$. For the root class, we use the variable $B_r$.
Constraint $F(op)$ states that if the solver selects an operator, it must also
select all its children; constraint $G(c)$ states that if the solver selects an
E-Class, it must select at least one of its members. Finally, we assert $B_r$ must 
be selected, which constrains the extracted expression
to be in the same E-Class as the unoptimized expression. These three constraints together
ensure the
selected nodes form a valid expression equivalent to the unoptimized input.
Satisfying these constraints, the solver now minimizes the cost function given
by the total cost of the selected operators. Because each $B_{op}$ represents an
operator node in the E-Graph which can be shared by multiple parents, this
encoding only assigns the cost once for every shared common subexpression. In
our implementation, we use Gurobi~\cite{gurobi} to solve the ILP problem.

\begin{figure}
    \begin{align*}
Constraints &\equiv B_r \wedge \bigwedge_{op} F(op) \wedge \bigwedge_c G(c)
\\ F(op) &\equiv B_{op} \rightarrow \bigwedge_{c \in op.children} B_c \\ G(c)
&\equiv B_c \rightarrow \bigvee_{op \in c.nodes} B_{op}
\end{align*}
\[\textbf{minimize} \sum_{op} B_{op} \cdot C_{op} \textbf{ s.t. } Constraints\]

    \caption{ILP constraint and objective for extraction. }
    \label{ilp}
\end{figure}

Each operation usually has cost proportional to the output size in terms of
memory allocation and computation. Since the size of a matrix is proportional
its the number of non-zeroes (nnz), we use SystemML's estimate of nnz as the
cost for each operation. Under our relational interpretation, this corresponds
to the cardinality of relational queries. We use the simple estimation scheme in
Figure~\ref{fig:cost}, which we find to work well. Future work can hinge on the
vast literature on sparsity and cardinality estimation to improve the cost
model.

\begin{figure}
\begin{align*}
    \textbf{S}[X*Y] &= min(\textbf{S}[X], \textbf{S}[Y]) \\ \textbf{S}[X+Y] &=
    min(1, \textbf{S}[X]+\textbf{S}[Y]) \\ \textbf{S}[\sum_i X] &= min(1, |i|
    \cdot \textbf{S}[X])
\end{align*}
    \caption[Sparsity estimation]{Sparsity estimation. We define $sparsity = nnz
      / size$, i.e. a 0 matrix has sparsity $0.0$\protect\footnotemark. $|i|$ is
      the size of the aggregated dimension. }
    \label{fig:cost}
\end{figure}
\footnotetext{Some may find this definition counter-intuitive; we define it so
  to be consistent with SystemML.}

\subsection{Schema and Sparsity as Class Invariant}
In the rules $R_{EQ}$ used by the saturation process, Rule~(\ref{RRC_ma}) If $i
\not\in A$, $A * \sum_i B = \sum_i (A * B)$ contains a condition on attribute
$i$ which may be deeply nested in the expression. This means the optimizer
cannot find a match with a simple pattern match. Fortunately, all expressions in
the same class must contain the same set of free attributes (attributes not
bound by aggregates). In other words, the set of free variables is invariant
under equality. This corresponds precisely to the schema of a database -
equivalent queries must share the same schema. We therefore annotate each class
with its schema, and also enable each equation to match on the schema.

In general, we find class invariants to be a powerful construct for programming
with E-Graphs. For each class we track as class invariant if there is a constant
scalar in the class. As soon as all the children of an operator are found to
contain constants, we can fold the operator with the constant it computes. This
seamlessly integrates constant folding with the rest of the rewrites. We also
treat sparsity as a class invariant and track it throughout equality saturation.
Because our sparsity estimation is conservative, equal expressions that use
different operators may have different estimates. But as soon as we identify
them as equal, we can merge their sparsity estimates by picking the tighter one,
thereby improving our cost function. Finally, we also take advantage of the
schema invariant during constraint generation. Because we are only interested in
RA expressions that can be translated to LA, we only generate symbolic variables
for classes that have no more than two attributes in their schema. This prunes
away a large number of invalid candidates and helps the solver avoid wasting
time on them. We implement class invariants using \texttt{egg}'s Metadata API.

\subsection{Translation, Operator Fusion and Custom Functions}
Since equality saturation can rewrite any expression given a set of equations,
we can directly perform the translation between LA and RA within saturation,
simply by adding the translation rules $R_{LR}$ from Figure~\ref{RMR}.
Furthermore, saturation has flexible support for custom functions. The simplest
option is to treat a custom functions as a black box, so saturation can still
optimize below and above them. With a little more effort, we have the option to
extend our equations $R_{EQ}$ to reason about custom functions, removing the
optimization barrier. We take this option for common operators that are not part
of the core RA semantics, e.g. square, minus and divide. In the best scenario,
if the custom function can be modeled by a combination of basic operators, we
can add a rule equating the two, and retain both versions in the same graph for
consideration. In fact, this last option enables us to encode fused operators
and seamlessly integrate fusion with other rewrite rules. As a result, the
compiler no longer need to struggle with ordering fusion and rewrites, because
saturation simultaneously considers all possible ordering.

\subsection{Saturation v.s. Heuristics}

Using equality saturation, SPORES elegantly remedies the drawbacks of
heuristics mentioned in the beginning of section~\ref{explore}. First, when two
or more conflicting rewrites apply, they would be added to the same E-Class, and
the extraction step will pick the more effective one based on the global cost
estimate. Second, there is no need to carefully order rewrites, because
saturation simultaneously considers all possible orders. For example, when rules
$R_1$ and $R_2$ can rewrite expression $e$ to either $R_1(R_2(e))$ or
$R_2(R_1(e))$, one iteration of saturation would add $R_1(e)$ and $R_2(e)$ to
the graph, and another iteration would add both $R_1(R_2(e))$ and $R_2(R_1(e))$
to the same E-Class. Third, rules do not need to reason about their dependency on
input or program properties, because extraction uses a global cost model that
holistically incorporates factors like input sparsity and common subexpressions.
Finally, every rule application in saturation applies one step of rewrite on top
of those already applied, naturally composing complex rewrites out of simple
ones.

\subsection{Integration within SystemML}

We integrate SPORES into SystemML to leverage its compiler
infrastructure. Figure~\ref{warp} shows the architecture of the integrated
system: the optimizer plugs into the algebraic rewrite pass in SystemML. It
takes in a DAG of linear algebra operations, and outputs the optimized DAG.
Within the optimizer, it first translates the LA DAG into relational algebra,
performs equality saturation, and finally translates the optimal expression back
into LA. We obtain matrix characteristics such as dimensions and sparsity
estimation from SystemML. Since we did not focus our efforts in supporting
various operators and data types unrelated to linear algebra computation (e.g.
string manipulation), we only invoke SPORES on important LA expressions
from the inner loops of the input program.

\begin{figure}
    \centering \includegraphics[width=0.6\linewidth]{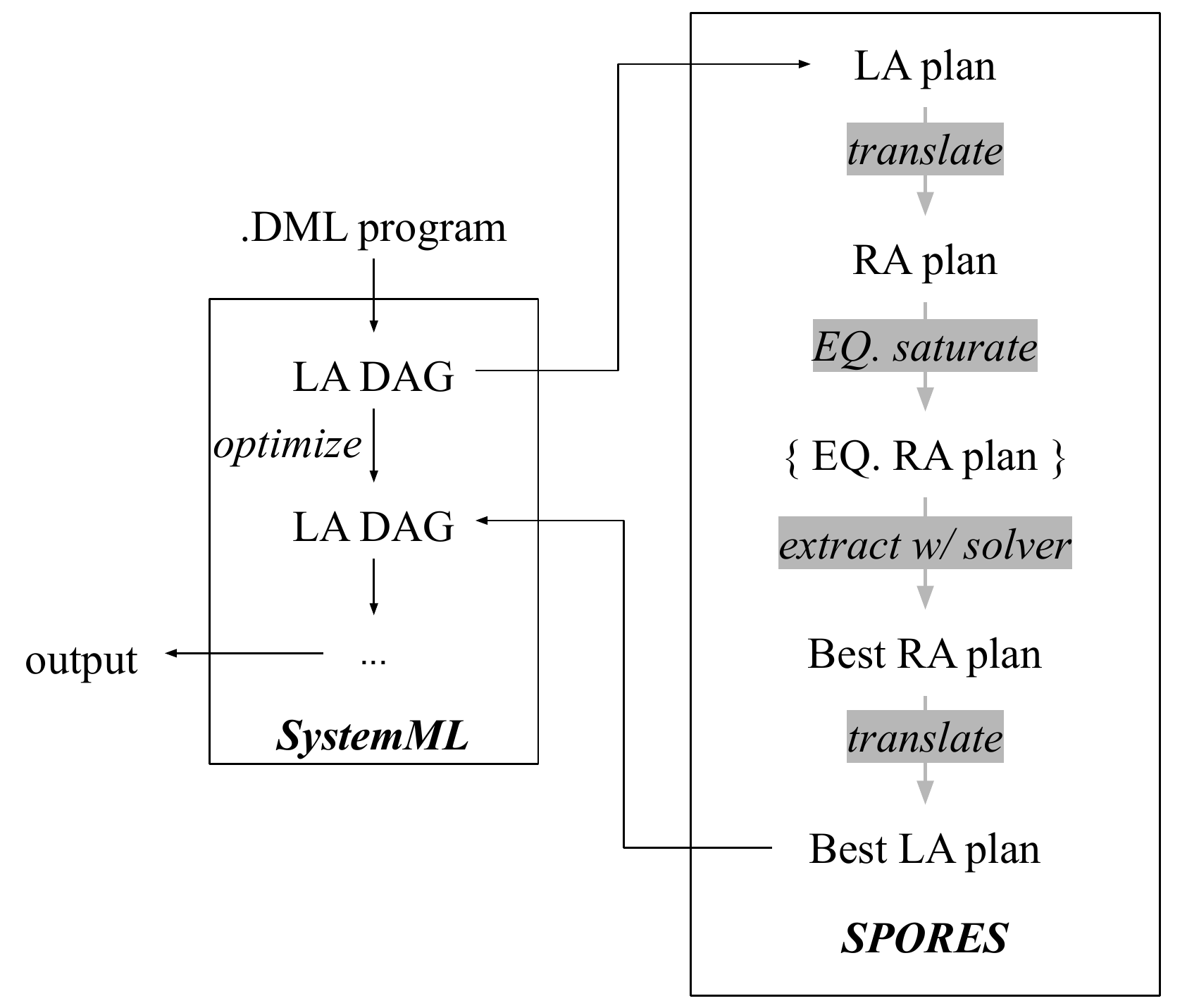}
    \caption{Architecture of SPORES and integration within SystemML}
    \label{warp}
\end{figure}{}

\section{Evaluation} \label{sec:evaluation}
We evaluate SPORES to answer three research questions about our
approach of relational equality saturation: 
\begin{itemize}
    \item \textbf{Section~\ref{relcomp}}: can SPORES derive hand-coded
      rewrite rules for sum-product optimization?
    \item \textbf{Section~\ref{perf}}: can SPORES find optimizations that
      lead to greater performance improvement than hand-coded rewrites and
      heuristics? 
      \item \textbf{Section~\ref{overhead}}: does SPORES induce compilation overhead afforded by its performance gain?  
\end{itemize}

We ran experiments on a single node with Intel E74890 v2 @ 2.80GHz with hyper-threading, 
1008 GB RAM, 8TB
disk, and Ubuntu 16.04.6. We used OpenJDK 1.8.0, Apache Hadoop 2.7.3, and Apache
Spark 2.4.4. Spark was configured to run locally with 6 executors, 8
cores/executor, 50GB driver memory, and 100GB executor memory. Our baselines are
from Apache SystemML 1.2.0. All datasets have been synthetically generated to
evaluate a range of scenarios, where we used algorithm specific data generators
from SystemML's benchmark suite.

\subsection{Completeness of Relational Rules} \label{relcomp}
% \begin{figure}
% \begin{enumerate}
%   \item $sum(A+B) = sum(A)+sum(B)$
%   \item $sum(X^T) = sum(X)$
%   \item $sum(sum_{row}(X)) = sum(X)$,  $sum_{col}$ similar
%   \item $(UV)^T = V^TU^T$, $(U*V)^T=U^T*V^T$
%   \item $(A+B)^T=A^T + B^T$, $(X^T)^T=X$
%   \item $(X*Y)*Z = X*(Y*Z)$, $X*Y = Y*X$
%   \item $A * X + b * X =  (A+b)*X$
%   \item $a * UV = (a*U) V$
% \end{enumerate}
% \caption{LA equality rules $R_{LA}$ for comparison. \textcolor{red}{TODO fix caption}}
% \label{RLA}
% \end{figure}

\begin{figure}
    \centering
    \begin{tabular}{|l|c|l|}
    \hline
     Method Name & \# & Example Rewrite \\
     \hline
     
\verb|UnnecessaryOuterProduct| & 3 & \verb|X*(Y%*%1) ->| \verb| X*Y, if Y col vector |\\

\verb|ColwiseAgg| & 3 & \verb|colsums(X) -> sum(X) or X, if col/row vector |\\

\verb|RowwiseAgg| & 3 & \verb|rowsums(X) -> sum(X) or X, if row/col vector |\\

\verb|ColSumsMVMult| & 1 & \verb|colSums(X*Y) -> t(Y) %*% X, if Y col vector |\\

\verb|RowSumsMVMult| & 1 & \verb|rowSums(X*Y) -> X %*% t(Y), if Y row vector |\\

\verb|UnnecessaryAggregate| & 9 & \verb|sum(X) -> as.scalar(X), if 1x1 dims |\\
\verb|EmptyAgg| & 3 & \verb|sum(X) -> 0, if nnz(X)==0 |\\
\verb|EmptyReorgOp| & 5 & \verb|t(X) -> matrix(0, ncol(X), nrow(X)) if nnz(X)==0|\\
\verb|EmptyMMult| & 1 & \verb|X%*%Y -> matrix(0,...), if nnz(Y)==0|\\
\verb|IdentityRepMatrixMult| & 1 & \verb|X%*%y -> X if y matrix(1,1,1) |\\
\verb|ScalarMatrixMult| & 2 & \verb|X%*%y -> X*as.scalar(y), if y is a 1-1 matrix |\\
\verb|pushdownSumOnAdd| & 2 & \verb|sum(A+B) -> sum(A)+sum(B) if dims(A)==dims(B) |\\

\verb|DotProductSum| & 2 & \verb|sum(v^2) -> t(v)%*%v if ncol(v)==1 |\\
\verb|reorderMinusMatrixMult| & 2 & \verb|(-t(X))%*%y -> -(t(X)%*%y) |\\

\verb|SumMatrixMult| & 3 & \verb|sum(A%*%B) -> sum(t(colSums(A))*rowSums(B)) |\\
\verb|EmptyBinaryOperation| & 3 & \verb|X*Y -> matrix(0,nrow(X),ncol(X)) / X+Y->X / X-Y->X |\\
\verb|ScalarMVBinaryOperation| & 1 & \verb|X*y -> X*as.scalar(y), if y is a 1-1 matrix |\\
\verb|UnnecessaryBinaryOperation| & 6 & \verb|X*1 -> X (after rm unnecessary vectorize) |\\
\verb|BinaryToUnaryOperation| & 3 & \verb|X*X -> X^2, X+X -> X*2, (X>0)-(X<0) -> sign(X) |\\
\verb|MatrixMultScalarAdd| & 2 & \verb|eps+U%*%t(V) -> U%*%t(V)+eps|\\
\verb|DistributiveBinaryOperation| & 4 & \verb|(X-Y*X) -> (1-Y)*X |\\
\verb|BushyBinaryOperation| & 3 & \verb|(X*(Y*(Z%*%v))) -> (X*Y)*(Z%*%v) |\\
\verb|UnaryAggReorgOperation| & 3 & \verb|sum(t(X)) -> sum(X) |\\
\verb|UnnecessaryAggregates| & 8 & \verb|sum(rowSums(X)) -> sum(X) |\\
\verb|BinaryMatrixScalarOperation| & 3 & \verb|as.scalar(X*s) -> as.scalar(X)*s |\\

\verb|pushdownUnaryAggTransposeOp| & 2 & \verb|colSums(t(X)) -> t(rowSums(X)) |\\
\verb|pushdownCSETransposeScalarOp| & 1 & \verb|a=t(X), b=t(X^2) -> a=t(X), b=t(X)^2 for CSE t(X) |\\

\verb|pushdownSumBinaryMult| & 2 & \verb|sum(lamda*X) -> lamda*sum(X) if lamdba is scalar|\\
\verb|UnnecessaryReorgOperation| & 2 & \verb|t(t(X))->X potentially introduced by other rewrites |\\
\verb|TransposeAggBinBinaryChains| & 2 & \verb|t(t(A)%*%t(B)+C) -> B%*%A+t(C) |\\
\verb|UnnecessaryMinus| & 1 & \verb|-(-X)->X potentially introduced by other rewrites |\\
\hline
\end{tabular}{}
    \caption{Sum-product rewrites in SystemML. The first column lists the name for each rewrite method. Each 
    method implements a number of rewrite patterns, and the second column shows how many. The last column shows an example rewrite for each method. Following SystemML's notation, scalars are in lower-case and matrices/vectors in upper-case. \texttt{\%*\%} is matrix multiply,  \texttt{t()} is transpose, and \texttt{nnz(X)} is the number of non-zeroes in $X$. 
    Equality saturation derives rewrites form all 31 methods (84 patterns) with relational  rules. }
    \label{rewritecomplete}
\end{figure}{}

Theoretically, our first hypothesis is validated by the fact that our relational
equality rules are complete w.r.t. linear algebra semantics. To test completeness in
practice\footnote{\textit{``I have only proved it correct, not tried it''} --
  Donald Knuth}, our first set of experiments check if SPORES
can derive the hand-coded sum product rewrite rules in SystemML. To do this, we
input the left hand side of each rule into SPORES, perform equality
saturation, then check if the rule's right hand side is present in the saturated
graph. The optimizer is able to derive
all 84 sum-product rewrite rules in SystemML using relational equality rules. 
See Figure~\ref{rewritecomplete} for a list of these rewrites. 
We believe replacing the 84 ad-hoc
rules with our translation rules $R_{LR}$ and equality rules $R_{EQ}$ would greatly
simplify SystemML's codebase. Together with equality saturation, our relational
rules can also lead to better performance, as we demonstrate in the next set of
 experiments.

% We then replace the RA equalities with the usual
% linear algebra equalities listed in Figure~\ref{RLA} and perform the same test.
% This time, the optimizer can only derive 33 out of 59 rules.
% TODO get number

% While we argue the experiment shows the benefit of RA rules over LA rules, one
% may wonder if we over-constrain the LA rules; perhaps 
% by adding a few additional rules, the LA rules can become complete as well. To
% this we answer that our RA rules can be viewed as an enhanced version of LA
% rules. While there may be other alternatives, we have shown our rules to be apt
% both in theory and in practice. From a different angle, the rules in
% Figure~\ref{rewritecomplete} from SystemML can also be regarded as an enhanced
% version of LA equalities. However, it is incomplete because it cannot derive 
% our example optimization of $sum((X-UV^T)^2)$. Practically, our RA rules are both more
% principled and simpler than the hand-coded rules, so adopting the former can
% simplify the codebase and alleviate maintenance burden. Together with equality saturation, our RA
% rules can also lead to better performance, as we demonstrate in the next set of
% experiments.

% TODO runtime evaluation 
\subsection{Run Time Measurement} \label{perf}

\begin{figure*}
    \centering  \includegraphics[width=\textwidth]{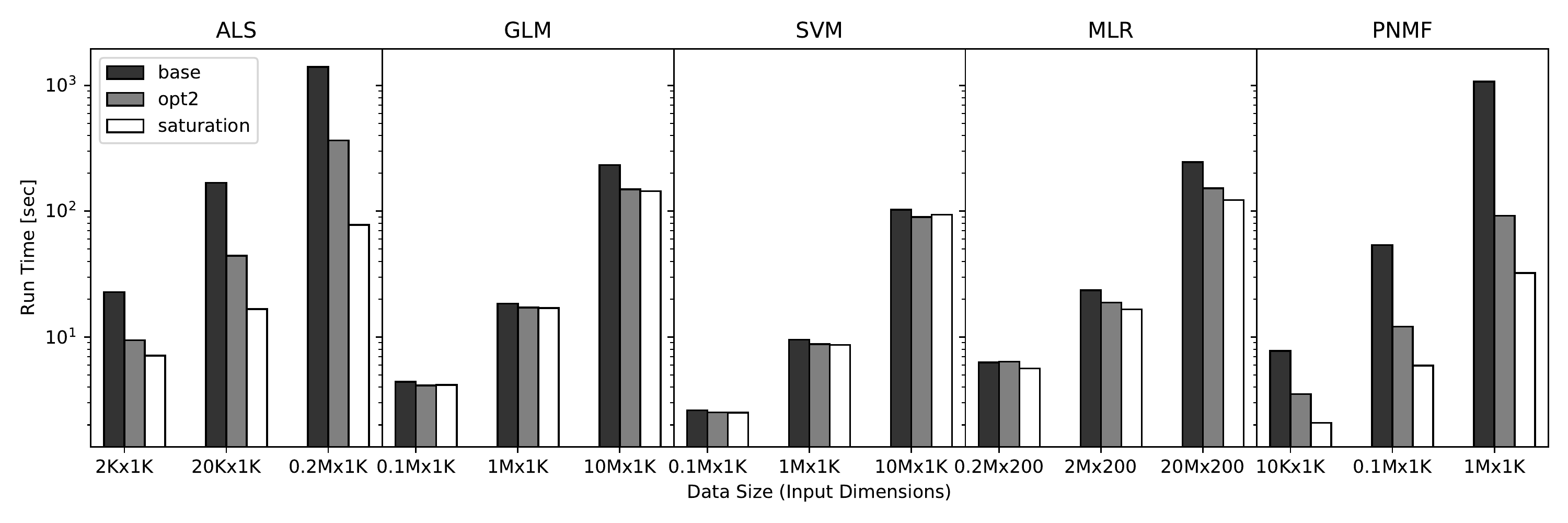}
    \caption{Run time of programs compiled by different optimizers. }
    \label{eval}
\end{figure*}
% \textcolor{red}{TODO: fix caption}

\begin{figure*}
    \includegraphics[width=\textwidth]{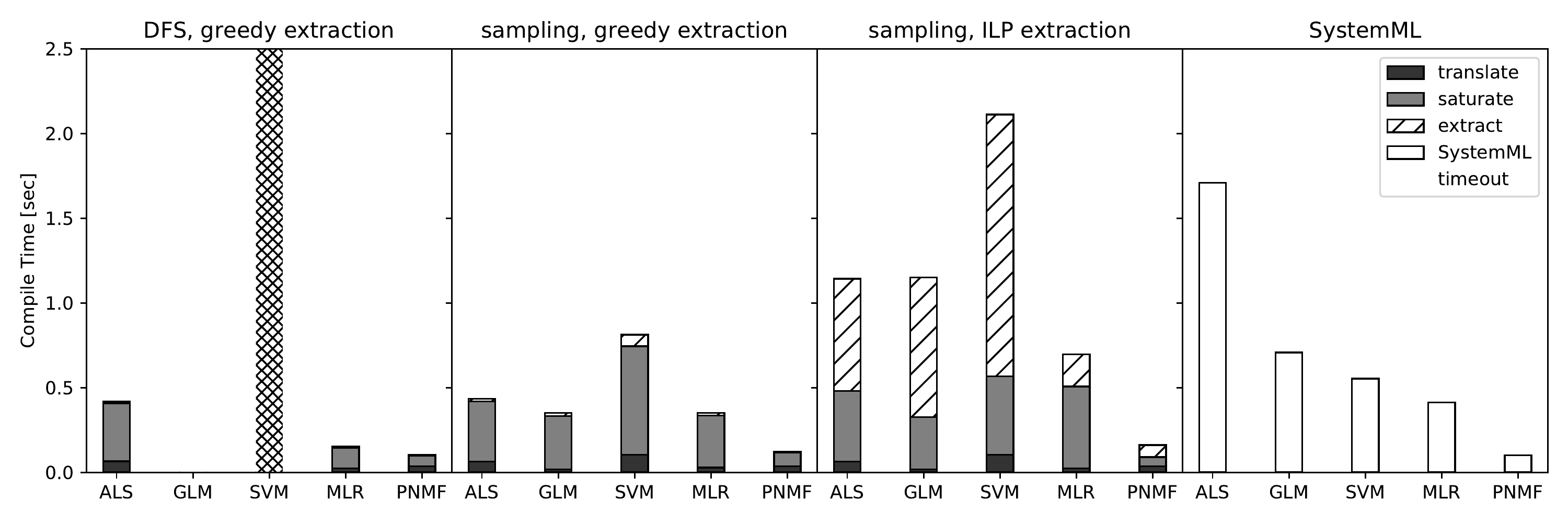}
    \caption{Compile time breakdown for different saturation and extraction
      strategies. Depth-first saturation reaches the 2.5s timeout compiling GLM and SVM. }
    \label{compile}
\end{figure*}{}
% {TODO: caption this}
We compare SPORES against SystemML's native optimizations for their
performance impact. In particular, we run SystemML with optimization level 2 (\texttt{opt2}),
which is its default and includes all advanced rewrites like constant folding
and common subexpression elimination. We additionally enable SystemML's native
sum-product rewrites and operator fusion. 
For baseline (\texttt{base}) we use level 1 optimization from SystemML, since level 0 (no optimization)
timeouts on almost all input sizes. \texttt{base} includes no advanced rewrites, 
sum-product optimization or operator fusion; it only performs local optimizations
with basic pattern-matching. 
We compile and execute 5 real-world
algorithms under 3 configurations:
\begin{enumerate*}
    \item SystemML without optimization,
    \item SystemML with optimization configured as above, and
    \item our equality saturation optimizer.
\end{enumerate*}{}
 The algorithms include Generalized Linear Model (GLM), Multinomial
 Logistic Regression (MLR), Support Vector Machine (SVM), Poisson
 Nonnegative Matrix Factorization (PNMF), and Alternating Least Square
 Factorization (ALS). We take the implementation of these algorithms from
 SystemML's performance benchmark suite.

Figure~\ref{eval} shows the performance improvement for each optimization
setting. Overall, equality saturation is competitive with the hand-coded rules
in SystemML: for GLM and SVM, saturation discovers the same
optimizations that improve performance as SystemML does. For ALS, MLR and PNMF,
saturation found new optimizations that lead to speedups from 1.2X to 5X
as compared to SystemML. We analyze each benchmark in
detail in the following paragraphs.

% \begin{figure}[t]
% \begin{lstlisting}[language=R]
% H=H*(t(W)%*%(X/(W%*%H+eps)))/t(colSums(W))
% W=W*((X/(W%*%H+eps))%*%t(H))/t(rowSums(H))
% obj=sum(W%*%H)-sum(X*log(W%*%H+eps))
% \end{lstlisting}
%     \caption{PNMF main loop. $X$ is a sparse matrix, $W$ and $H$ are low-rank matrices, and
%       $eps$ is a scalar.}
%     \label{pnmf}
% \end{figure}
% \begin{figure}[t]
% \begin{lstlisting}[language=R]
% Q=P[,1:K] * (X %*% ssX_V)
% HV=t(X)%*%(Q - P[,1:K]
%       *(rowSums(Q)%*%matrix(1,1,K)))
% \end{lstlisting}
%     \caption{Part of MLogReg main loop. $X$ and $ssX\_V$ are matrices. In our experiment $K=1$ so
%       $P[,1:K]$ and $Q$ are column vectors, and $matrix(1,1,K)$ is a 1x1
%       matrix holding the value $1$. }
%     \label{mlogreg}
% \end{figure}

% TODO no section 
% \subsection{Rewrite Details} \label{optdetails}

For \textbf{ALS}, SPORES leads to up to 5X speedup beyond SystemML's
optimizations using our relational rules. Investigating the optimized code
reveals the speedup comes from a rather simple optimization: SPORES expands
$(UV^T - X)V$ to $UV^TV-XV$ to exploit the sparsity in $X$. Before the optimization,
all three operations (2 matrix multiply and 1 minus) 
in the expression create dense intermediates because $U$
and $V$ are dense. After the optimization, $XV$ can be computed efficiently thanks to
the sparsity in $X$. $UV^TV$ can be computed in one go without intermediates, taking advantage of
SystemML's \texttt{mmchain} operator for matrix multiply chains. Although the
optimization is straightforward, it is counter-intuitive because one expects
computing $A(B + C)$ is more efficient than $AB + AC$ if one does not consider
sparsity. For the same reason, SystemML simply does not consider distributing
the multiplication and misses the optimization.

For \textbf{PNMF}, the speedup of up to 3X using RA rules attributes to rewriting $sum(WH)$
to $sum_{col}(W) \cdot sum_{row}(H)$ which avoids materializing a dense
intermediate $WH$. Interestingly, SystemML includes this rewrite rule but did
not apply it during optimization. In fact, SystemML only applies the rule when
$WH$ does not appear elsewhere, in order to preserve common subexpression.
However, although $WH$ is shared by another expression in PNMF, the other
expression can also be optimized away by another rule. Because both rules uses heuristics to favor sharing CSE, neither fires. 
This precisely
demonstrates the limitation of heuristics. 

For \textbf{MLR}, the important optimization\footnote{Simplified here for
  presentation. In the source code $P$ and $X$ are not variables but consist of
  subexpressions. } by saturation is $P * X - P * sum_{row}(P) * X$ to
$P*(1-P)*X$, where $P$ is a column vector. This takes advantage of the
\texttt{sprop} fused operator in SystemML to compute $P*(1-P)$, therefore
allocating only one intermediate. Note that the optimization factors out $P$, which
is the exact opposite to the optimization in ALS that distributes multiply. Naive
rewrite rules would have to choose between the two directions, or resort to
heuristics to decide when to pick which. 

% TODO check dimensions
% \begin{figure}[t]
% \begin{lstlisting}[language=R]
% out=1-Y*(Xw+step_sz*Xd);
% sv=(out>0);
% out=out*sv;
% g=wd+step_sz*dd-sum(out*Y*Xd); 
% h=dd+sum(Xd*sv*Xd);
% step_sz = step_sz - g/h;
% \end{lstlisting}
%     \caption{Part of L2SVM main loop. }
%     \label{l2svm}
% \end{figure}
% \begin{figure}[t]
% \begin{lstlisting}[language=R]
% t_gp=1.0/(1.0+abs(flt)*0.231641888)
% pt_gp = t_gp * ( 0.254829592 
%       + t_gp * (-0.284496736 
%       + t_gp * ( 1.421413741 
%       + t_gp * (-1.453152027 
%       + t_gp *   1.061405429))));
% \end{lstlisting}
%     \caption[gauss error function]{
%     Part of GLM main loop (approximates the Gauss error function~\cite{abramowitz+stegun}). $t\_gp$ is a matrix.}
%     \label{glm}
% \end{figure}

For \textbf{SVM} and \textbf{GLM}, equality saturation finds the same optimizations as
SystemML does, leading to speedup mainly due to operator fusion. Upon
inspection, we could not identify better optimizations for \textbf{SVM}. For
\textbf{GLM}, however, we discovered a manual optimization that should improve
performance in theory, but did not have an effect in practice since SystemML
cannot accurately estimate sparsity to inform execution. 

\subsection{Compilation Overhead} \label{overhead}
\begin{figure*}
    \centering  \includegraphics[width=\textwidth]{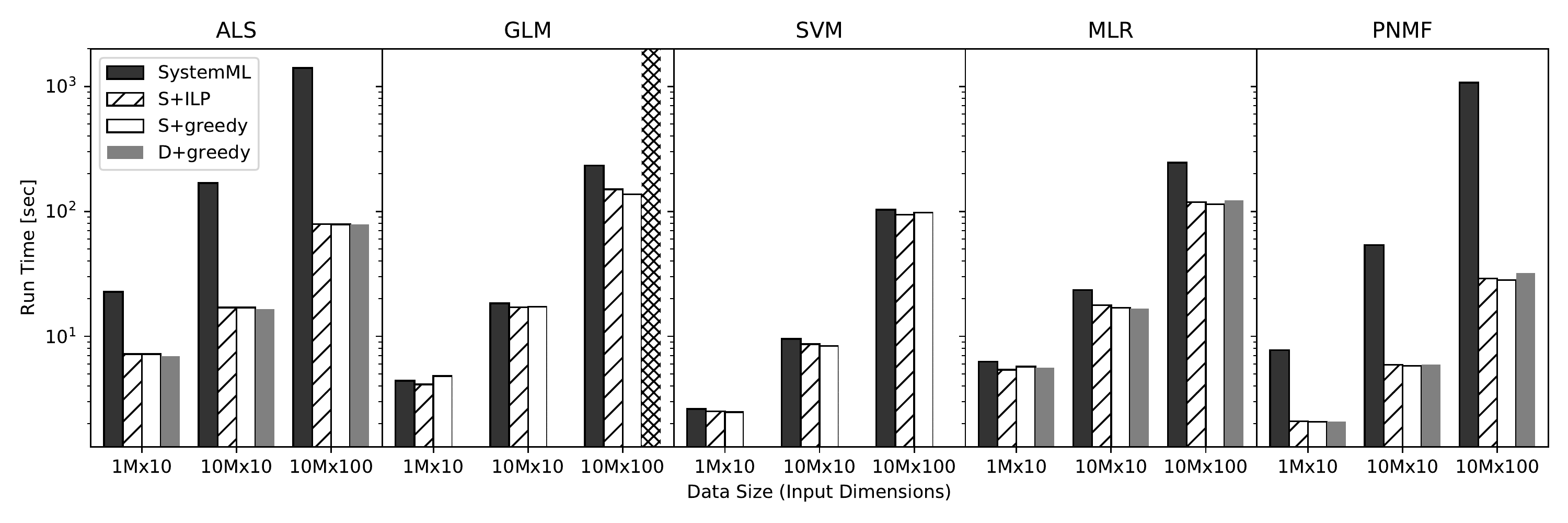}
    \caption{Performance impact of different saturation and extraction
      strategies. S is saturation with sampling, and D is depth-first saturation. Depth-first saturation runs into timeout compiling GLM and SVM. }
    \label{sampleeval}
\end{figure*}
% TODO find number \textcolor{red}{TODO fix caption}}
In our initial experiments, SPORES induces nontrivial compilation
overhead compared to SystemML's native rule-based rewrites. Figure~\ref{compile}
 (sampling, ILP extraction) shows the compile time breakdown for each benchmark, and the majority of time is
spent in the ILP solver. We therefore experiment with a greedy algorithm during
extraction to see if we can trade off guarantees of optimality for a shorter
compile time. This algorithm traverses the saturated graph bottom-up, picking
the cheapest operator in each class at every level. Figure~\ref{sampleeval}
shows the performance impact of greedy extraction, and Figure~\ref{compile} (sampling, greedy extraction)
shows the compile time with it. Greedy extraction significantly reduces compile
time without sacrificing \textit{any} performance gain! This is not surprising
in light of the optimizations we discussed in Section~\ref{perf}: all of
these optimizations improve performance regardless of common
subexpressions, so they are selected by both the ILP-based and the greedy
extractor.

We also compare saturation with sampling against depth-first saturation in terms of
performance impact and compile time. Recall the depth-first saturation strategy
applies all matches per rule per iteration. As Figure~\ref{compile} shows,
sampling is slightly slower for ALS, MLR and PNMF, but resolves the timeout for 
GLM and SVM. This is because sampling takes longer to converge when full
saturation is possible, and otherwise prevents the graph from blowing up before
reaching the iteration limit.
Indeed, saturation converges for ALS, MLR and PNMF, which means SPORES
can guarantee the optimality of its result under the given cost model. Saturation
does not not converge before reaching the iteration limit for GLM and SVM because
of deeply nested $*$ and $+$ in the programs. 
Convergence
may come as a surprise despite E-Graph's compaction -- expansive rules like
associativity and commutativity commonly apply in practice. However, the expression
DAGs we encounter are often small (no more than 15 operators), and large DAGs
are cut into small pieces by optimization barriers like uninterpreted functions. 

Figure~\ref{compile} compares the overall DAG compilation overhead of SystemML
against SPORES with different extraction strategies.
Note that the overhead of SystemML also includes optimizations unrelated to sum-product rewrites that are difficult to disentangle, therefore it only gives a sense of the base compilation time and does not serve as head-to-head comparison against SPORES. 
Although SPORES induces significant compilation overhead in light of SystemML's total DAG compilation time, the overhead is afforded by
its performance gain. As we did not focus our efforts on
reducing compile time, we believe there is plenty room for improvement, for
example organizing rewrite rules to speed up saturation.

\section{Related Work}
% TODO emphasize SystemML
% TODO ADD Immanual Trummer, PILP
% Morpheus
% lara
% volcano
There is a vast body of literature for both relational query optimization and
optimizing compilers for machine learning. Since we optimize machine learning programs through a
relational lens, our work relates to research in both fields. As we have pointed
out, numerous state-of-the-art optimizing compilers for machine learning
 resort to syntactic rewrites and heuristics to optimize linear algebra
expressions~\cite{DBLP:reference/bdt/Boehm19}~\cite{DBLP:conf/icml/SujeethLBRCWAOO11}~\cite{DBLP:conf/sigmod/HuangB013}. We distinguish our work which performs optimization based on a
relational semantics of linear algebra and holistically explore the complex
search space. A majority of relational query optimization focus on join order
optimization \cite{Graefe95a} \cite{MoerkotteN06} \cite{MoerkotteN08}
\cite{selinger1979access}; we distinguish our work which optimizes programs with
join (product), union (addition), and aggregate (sum) operations. Sum-product optimization considers operators other than join while optimizing
relational queries. Recent years have seen a line of excellent theoretical and
practical research in this area \cite{KhamisNR16} \cite{Joglekar2016AJARAA}.
These work gives significant improvement for queries involving $*$ and $\sum$,
but fall short of LA workloads that occur in practice. We step past these
frameworks by incorporating common subexpressions and incorporating addition
($+$). 

In terms of approach, our design of relational IR ties in to research that
explores the connection between linear algebra and relational algebra. Our
design and implementation of the optimizer ties into research that leverage
equality saturation and \textsc{and-or dag}s for query optimization and compiler
optimization for programming languages. Since we focus on optimizing sum-product
expressions in linear algebra, our work naturally relates to research in sum-product
optimization. We now discuss these three lines of research in detail. 

\subsection{Relational Algebra and Linear Algebra}
Elgamal et al. \cite{ElgamalLBETRS17} envisions \textsc{spoof}, a compiler for machine learning programs
that leverages relational query optimization techniques for LA sum-product optimization. 
We realize this vision by providing the translation
rules from LA to RA and the relational equality rules that completely represents
the search space for sum-product expressions. One important distinction is, Elgamal
et al. proposes \emph{restricted relational algebra} where every expression must
have at most two free attributes. This ensures every relational expression in every step of the
optimization process to be expressible in LA. In contrast, we remove this restriction and only require the 
optimized output to be in linear algebra. This allows us to trek out to spaces not
covered by linear algebra equality rules and achieve completeness. In addition 
to sum-product expressions, Elgamal et al. also considers selection and projection operations
like selecting the positive entries of a matrix. We plan to explore supporting
selection and projection in the future. Elgamal et al. also proposes compile-time generation of
fused operators, which is implemented by Boehm et al.~\cite{DBLP:journals/pvldb/BoehmRHSEP18}.  SPORES readily takes advantage of existing fused operators, and we
plan to explore combining sum-product rewrite with fusion generation in the future. 

MorpheusFI by Li et al.~\cite{DBLP:conf/sigmod/LiC019} and LARA by Hutchison et al.~\cite{DBLP:journals/corr/HutchisonHS17} explore optimizations across the interface of machine learning and database systems. In particular, MorpheusFI speeds up machine learning algorithms over large joins by pushing computation into each joined table, thereby avoiding expensive
materialization. LARA implements linear algebra operations with relational operations and
shows competitive optimizations alongside popular data processing systems. 
Schleich et al.\cite{DBLP:conf/sigmod/SchleichOC16} and Khamis et al.\cite{DBLP:journals/corr/NgoNOS17} explore in-database learning, which aims to push entire
machine learning algorithms into the database system. 
We contribute in this space by showing that even without a relational engine, the
relational abstraction can still benefit machine learning tasks as a powerful
intermediate abstraction. 

\subsection{Equality Saturation and AND-OR DAGs}
Equality saturation and \textsc{and-or dag}s have been applied to optimize
low-level assembly code~\cite{DBLP:conf/pldi/JoshiNR02}, Java programs \cite{DBLP:journals/corr/abs-1012-1802}, 
database queries \cite{Graefe95a}, floating point arithmetics \cite{DBLP:conf/pldi/PanchekhaSWT15}, and even computer-aided
design models \cite{DBLP:journals/corr/abs-1909-12252}. The design of our relational IR brings unique challenges
in adopting equality saturation. Compared to database query optimizers 
that focus on optimizing join orders, unions and aggregates play a central
role in our relational IR and are prevalent in real-world programs. As
a result, our
equality rules depend on the expression schema which is not immediately accessible
from the syntax. We propose class invariants as a solution to access schema
information, and show it to be a powerful construct that enables constant
folding and improves cost estimation. Compared to optimizers for low-level
assembly code or Java program, we commonly encounter linear algebra expressions
that trigger expansive rules and make saturation convergence impractical. 
We propose to sample rewrite matches in order to achieve good exploration
without full saturation. Equality saturation takes advantage of constraint solvers which have also been applied to program optimization and query
optimization. In particular, the use of solvers for satisfiability
modulo theories by \cite{DBLP:conf/asplos/Solar-LezamaTBSS06} has spawned a paradigm now known as \textit{program synthesis}. 
In query optimization research, 
\cite{DBLP:conf/sigmod/Trummer017} applies Mixed Integer Linear Programming for optimizing join ordering. 
Although constraint solvers offer pleasing guarantees of optimality, our
experiments show their overhead does not bring significant gains for optimizing
LA expressions. 

% \subsection{Exploiting Sparsity}
% A number of optimizing compilers for linear algebra programs support sparse
% matrix operators . SystemML \cite{boehm2014systemml} and Cumulon
% \cite{DBLP:conf/sigmod/HuangB013} includes a variety of hand-fused sparse operators.
% \cite{ElgamalLBETRS17} further automates the generation of fused operator
% implementation. However, all existing approach either work on a per-operator
% basis, or rely on hand-coded rules to rewrite expressions to fused operators.
% Using equality saturation, we can automatically discover opportunity for
% operator fusion based on a simple set of equality rules. And even with the
% absence of operator fusion, we can still take advantage of sparsity in our
% optimization.

% \subsection{Multi-query Optimization}
% The research in Multi-query Optimization aims to share common subqueries among
% multiple queries to speed up processing \cite{Kathuria017}. This is closely
% related to our problem of exploiting CSEs in linear algebra, although our focus
% is intra-query instead of inter-query. Moreover, the majority of multi-query
% optimizers focus on joins, whereas union and aggregate are essential for linear
% algebra programs. One lesson we can learn from the literature multi-query
% optimization is to perform equality saturation for multiple expressions in the
% same egraph, thereby enabling sharing across optimization sessions and reducing
% compilation time.

% \section{Future Work}
% \input{future.tex}

\section{Conclusion}

We propose a novel optimization approach for compiling linear algebra programs,
using relational algebra as an intermediate representation and equality
saturation to explore the search space. We implement our equality saturation
based optimizer and show it is effective for improving real-world machine
learning algorithms.

\section{Acknowledgement}
The authors would like to thank Alexandre Evfimievski, Matthias Boehm, and Berthold Reinwald for their insights in developing SystemMLs internals. We would also like to thank Mathew Luo and Brendan Murphy for their guidance in piloting an ILP extractor as well as Zach Tatlock, Max Willsey and Chandrakana Nandi for their valuable feedback and discussion.

\balance

\bibliographystyle{abbrv}
\bibliography{vldb2020}

\newpage

\appendix
\section{Uniqueness of RA Canonical Form}

To prove Lemma~\ref{lemma:unique:nf} (\textit{uniqueness of canonical form}), we first give formal definitions for several important constructs. First, we interpret a tensor (high dimensional matrix) as a function from a tuple of indices to a real number: 
\begin{definition}{{\em (Tensors)}}
A {\em \textbf{tensor}} is a function $A: [N]^d \rightarrow \R$ such that $A(i,j,\dots)$ returns the tensor entry $A_{ij\dots}$. $d$ is the {\em \textbf{dimensionality}} of $A$, and $N$ is the {\em \textbf{dimension size}}. Each argument to $A$ is an {\em \textbf{index}}, and we write \pmb{$i$} (in bold) for a tuple of indices. We write $\pmb{A}$ (in bold) for $A(\pmb{i})$, i.e. a tensor indexed with a tuple of indices. 
\end{definition}
We use the terms {\em \textbf{tensor}} and {\em \textbf{relation}} interchangingly. Note that we assume every dimension has the same size $N$ for simplicity. If the dimension sizes differ in a tensor $A$, we can easily set $N$ to be the maximum size and fill $A$ with zeroes accordingly. For example, for a 2D matrix $A$ with dimensions $|d_1| = m < |d_2| =n$, we simply set $N=n$ and have $A_{ij}=0$ for $i>m$. 

Now we give names for special forms of expressions at each level of the normal form: 
\begin{definition}{{\em (Atoms, Monomials, Terms, Polyterms)}}\label{forms} 
An indexed tensor $\pmb{A}$ is also called an {\em \textbf{atom}}. 
A {\em \textbf{monomial}} $m$ is a product of any number of atoms. A {\em \textbf{term}} $t$ is an aggregate of a monomial over a set of indices. A {\em \textbf{polyterm}} $e$ is the sum of terms (with a constant term $c$ at the end), where each term has a constant factor: 
\begin{alignat*}{3}
      &\pmb{A} & & := A(\pmb{i}) && (1)\, atoms \\
      &m & & := \pmb{A}_1 \times \pmb{A}_2 \times \dots \pmb{A}_n && (2)\, monomials  \\
      &t & & := \textstyle \sum_{\pmb{i}} m  &&(3)\, terms \\
      &e & & := c_1 t_1 + c_2 t_2 + \dots c_n t_n + c \,&&(4)\, polyterms
\end{alignat*}{}
We identify the monomial $m$ with a bag of atoms, denoted $bag(m)$, such
  that for every atom $\pmb{A}$ occurring $n$ times in $m$, $bag(m)$
  contains $n$ copies of $\pmb{A}$. 
  An index $i$ in a term $t=\sum_{\pmb{i}} m$ is {\em \textbf{bound}} if $i \in \pmb{i}$; otherwise it is {\em \textbf{free}}. We write $bv(t)$ for the set of bound indices in $t$, $fv(t)$ for the set of free indices in $t$, and $vars(t)$ for the set of all indices in $t$. 
\end{definition}{}
\begin{example}{}
The polyterm 
$2 \sum_{j} A(i,j) \times A(i,j) \times B(j,k) \times B(j,k) + 3 \sum_{l} A(i,l) \times C(l,k) + 2$ represents the linear algebra expression $2A^2B^2 + 3AC + 2$, where $A^2$ squares the matrix element-wise. The monomial $A(i,j) \times A(i,j) \times B(j,k) \times B(j,k)$ is the same as $A(i,j)  \times B(j,k) \times A(i,j) \times B(j,k)$, and we view it as the bag $\{A(i,j), A(i,j), B(j,k), B(j,k)\}$. 
\end{example}
Before giving the formal definition of a canonical form, we need to define two syntactical relationships between our expressions, namely {\em homomorphism} and {\em isomorphism}. Fix terms $t = \sum_{\pmb{i}} m$ and $t' = \sum_{\pmb{i'}} m'$, and let $f: \pmb{i} \rightarrow \pmb{i}'$ be any function. Let $\pmb{A} \in bag(m)$ be an atom of $m$. We write $f(A)$ for the result of applying $f$ to all bound indices of $\pmb{A}$. We write $f(bag(m))$ for the bag obtained by applying $f$ to each atom $\pmb{A} \in bag(m)$. 
\begin{definition}{\em (Term Homomorphism)}
A homomorphism $f: t \rightarrow t'$ is a function $f: \pmb{i} \rightarrow \pmb{i}'$ such that $f(bag(m))=bag(m')$. Note that $f$ is a one-to-one mapping between $bag(m)$ and $bag(m')$. 
\end{definition}{}
\begin{example}
Between terms
$t_1 = \sum_{vwst} A(i,v) \times B(v,w) \times A(i,s) \times B(s,t)$ and 
$t_2= \sum_{jk} A^2(i,j) \times B^2(j,k)$
there is a homomorphism:  $$t_1 \rightarrow t_2:[v \mapsto j, w \mapsto k, s \mapsto j, t \mapsto k]$$
\end{example}
We extend $f$ to take $vars(t_1)$ where free variables map to themselves. It is easy to see for any homomorphism $f$, $f(vars(t_1))=vars(t_2)$. 
\begin{corollary}\label{surjective}
Homomorphism is \emph{\textbf{surjective}} on the indices. 
\end{corollary}
\begin{proof}
Suppose for the sake of contradiction that a homomorphism $f: t_1 \rightarrow t_2$ is not 
surjective. Then there exists an index $i \in vars(t_2)$ that is not in $f(vars(t_1))$, 
and the atom containing $i$ does not appear in $f(t_1)$. That implies the monomial in $f(t_1)$ is cannot be equal to the monomial in $t_2$, so $f$ is not a homomorphism -- contradiction. 
\end{proof}{}

\begin{corollary}\label{compose}
Homomorphism is closed under composition: 
given homomorphisms $f: t_1 \rightarrow t_2$ and $g: t_2 \rightarrow t_3$, $g \circ f: t_1 \rightarrow t_3$ is a homomorphism from $t_1$ to $t_3$. 
\end{corollary}
\begin{proof}
A homomorphism is a function on indices, so composing homomorphisms is just composing functions. 
\end{proof}{}
A stronger correspondence between terms is an \emph{isomorphism}: 
\begin{definition}[Term Isomorphism]
Terms $t_1=\sum_{\pmb{i_1}}m_1$ and $t_2=\sum_{\pmb{i_2}}m_2$ are \emph{\textbf{isomorphic}} iff there is a bijective homomorphism between them. 
We write $t_1 \equiv t_2$ to mean $t_1$ and $t_2$ are isomorphic. 
\end{definition}

A pair of homomorphisms $t_1 \rightarrow t_2$ and $t_2 \rightarrow t_1$ produce an isomorphism: 
\begin{lemma}\label{homocycle}
Given two terms $t_1$ and $t_2$, if there is a homomorphism $f: t_1 \rightarrow t_2$ and a homomorphism $g: t_2 \rightarrow t_1$ then $t_1 \equiv t_2$. If there is a cycle of homomorphisms among a number of terms, all terms on the cycle are isomorphic. 
\end{lemma}{}
\begin{proof}
By Corollary~\ref{surjective}, a pair of homomorphisms between two terms are a pair of surjective maps between the terms' indices. A pair of surjective maps induce a bijective map which is an isomorphism. 
By Corollary~\ref{compose}, if $t_1$ and $t_2$ are on a cycle of homomorphism, we can retract the homomorphism chains between them to obtain a pair of homomorphisms $f: t_1 \rightarrow t_2$ and $g: t_2 \rightarrow t_1$, which implies $t_1 \equiv t_2$. 
\end{proof}{}
We are now ready to formally define the canonical form for RA expressions: 

\begin{definition}{\em (Canonical Form)} An RPlan expression (as defined in Table~\ref{tPlanOps}) is {\em \textbf{canonical}} if it is a polyterm (Definition~\ref{forms}) containing no isomorphic terms. 
\end{definition}{}

Our ultimate goal is to identify canonical form isomorphism with equivalence. That is, two canonical expressions are equivalent iff they are isomorphic. By equivalence we mean the expressions evaluate to the same result given any same inputs: 

\begin{definition}{\em (Equivalence of Expressions)} Two expressions are {\em \textbf{equivalent}} iff $\forall \pmb{T}. e_1(\pmb{T}) = e_2(\pmb{T})$, where $\pmb{T}$ is the set of input tensors to the expressions. We write $e_1 = e_2$ to mean $e_1$ and $e_2$ are equivalent. 
\end{definition}{}

We can canonicalize any expression by pulling $+$ to the top and pushing $\times$ to the bottom, while combining isomorphic terms $c_1 t + c_2 t$ into $(c_1 + c_2) t$: 

\begin{lemma}
For every RPlan expression, there is an equivalent canonical expression. 
\end{lemma}{}
\begin{proof}
The proof is a standard application of the rewrite rules $R_{EQ}$ in Figure~\ref{RRC}. 
\end{proof}{}

We can identify canonical expressions syntactically using term isomorphism: 

\begin{definition}{\em (Isomorphic Canonical Expressions)} 
Given two canonical expressions $e = \sum_{\pmb{i}} c_1 t_1 + \dots + c_n t_n + c $ and $e' = \sum_{\pmb{i'}} c_1' t_1' + \dots + c_m' t_m' + c'$, $e$ and $e'$ are isomorphic if there is a bijection $\sigma: [n] \rightarrow [m]$ (in particular $n=m$), such that $\forall i \in [n]$, 
     $c_i = c_{\sigma(i)}'$, 
     $t_i \equiv t_{\sigma(i)}'$, and 
     $c = c'$. 
\end{definition}{}

Note that isomorphic expressions have the same free variables. We now show two expressions are isomorphic iff they are equivalent: 

\begin{theorem}{(\textbf{Isomorphism Captures Equivalence})} For canonical expressions $e_1$ and $e_2$: 
$$e_1 \equiv e_2 \iff e_1 = e_2$$ 
\end{theorem}{}
\begin{proof}
The left-to-right direction is straightforward: isomorphism only renames indices and reorders $+$ and $\times$, which does not change semantics. And because RA semantics is deterministic, two isomorphic expressions compute the same function. 

The right-to-left direction is exactly Lemma~\ref{lemma:unique:nf}, the uniqueness of canonical forms: 
$$e_1  = e_2  \Rightarrow e_1  \equiv e_2 $$
Without loss of generality, we make the following simplifying assumptions. First, we assume $e_1$ and $e_2$ contain no constant terms. Otherwise the expressions would evaluate to their respective constants on all-0 inputs, so the constant terms must be equal. Subtracting the same constant preserves equivalence and isomorphism, so we can simply remove equal constant terms. 
Second, we assume no term in $e_1$ is isomorphic to any term in $e_2$. 
Otherwise if $e_1$ contains $c_1t_1$ and $e_2$ contains $c_2t_2$ with $t_1 \equiv t_2$, then we can write $e_1 = c_1 t_1 + e_1'$ and $e_2 = c_2 t_2 + e_2'$ where $e_1'$ and $e_2'$ are polyterms. Assume w.l.o.g that $c_1 \leq c_2$. Since $c_1t_1 + e_1' = c_2 t_2 + e_2'$, $e_1' = (c_2 - c_1)t_2 +e_2'$ and we have removed $t_1$ from $e_1$. 
Furthermore, $e_1 = e_2 \Rightarrow e_1' = e_2'$ and $e_1' \equiv e_2' \Rightarrow e_1 \equiv e_2$. 
Then by proving $e_1' = e_2' \Rightarrow  e_1' \equiv e_2'$ we can show the following:
$$e_1 = e_2 \Rightarrow e_1' = e_2' \Rightarrow  e_1' \equiv e_2' \Rightarrow e_1 \equiv e_2$$
Finally we assume $e_1$ and $e_2$ are fully aggregated, i.e. every index is bound. Without this assumption, we first observe $e_1$ and $e_2$ must have the same free indices to output tensors with the same dimensionality. Then we define $e_1' = \sum_{\pmb{i}} e_2$ and $e_2' = \sum_{\pmb{i}} e_2$, where $\pmb{i}$ is the set of free variables of $e_1$ and $e_2$. We can easily show $e_1 = e_2 \Rightarrow e_1' = e_2'$ and $e_1' \equiv e_2' \Rightarrow e_1 \equiv e_2$, and with a proof of $e_1' = e_2' \Rightarrow e_1' \equiv e_2'$ it follows $e_1 = e_2 \Rightarrow e_1 \equiv e_2$. 

Each expression can now be viewed as a set of terms, where each term has a constant factor and no two terms are isomorphic. Denoting by $t_1 < t_2$ if there is a homomorphism $t_1 \rightarrow t_2$, we observe homomorphism induces a partial order on the terms of $e_1$ and $e_2$. There is no cycle of homomorphisms, because by Lemma~\ref{homocycle} such a cycle implies isomorphisms, but we have assumed no isomorphic terms. 

We now prove Lemma~\ref{lemma:unique:nf} through its contrapositive: 

$$e_1 \not \equiv e_2 \Rightarrow e_1 \not = e_2$$

We proceed by constructing a set of witness input tensors given which $e_1$ and $e_2$ return different results. First, let $t_1 = \sum_{i_1 i_2 \dots} \pmb{X}_1 \times \pmb{X}_2 \times \dots$ be the minimal term under the partial order induced by homomorphism. Assume w.l.o.g $t_1 \in e_1$. Define bijective function $\phi: i_n \rightarrow n$ mapping each index to a unique number. For each $\pmb{X}_n = X_n(\pmb{i}_n)$ introduce a real-valued variable $x_{n}$ and construct the input tensor $X_n$ such that $X_{\pmb{i}_n} = x_n$. Finally, set all undefined entries of each input tensor to $0$. 

Given the input tensors defined above, $t_1$ evaluates to a polynomial $p_1$ of variables $x_1 \dots x_n$. Every monomial in $p_1$ corresponds to a function from $t_1$'s atoms $\pmb{X}_a, \pmb{X}_b, \dots$ to some variables $x_m, x_n, \dots$, and we call such a function $f_m$. There is one special (bijective) $f_m$ that maps each atom to its own variable: 
let $m_1$ be the monomial in $t_1$, $m_1 = \pmb{X}_1 \times \pmb{X}_2 \times \dots$. We define $f_p(t_1) = m_1(\phi(i_1), \phi(i_2), \dots) = x_1 x_2 \dots$. $f_p(t_1)$ does not appear in any polynomial from other terms. Otherwise if it appears in polynomial $p_2$ from term $t_2$, we would be able to construct a homomorphism
 $t_2 \rightarrow t_1 = f_m \circ f_p^{-1}$, where $f_m$ maps the monomial $m_2$ in $p_2$ to $t_1$. But that is impossible because we have picked $t_1$ to be the minimal term under homomorphism. Therefore, $p_1$ differs from any polynomial from terms in $t_2$, and by the fundamental theorem of algebra the polynomial are inequivalent. As a result, $e_1 \not = e_2$. 
\end{proof}{}

\end{document}